\documentclass[12pt, journal, draftclsnofoot, onecolumn]{IEEEtran}

\usepackage{amsthm,amssymb,graphicx,multirow,amsmath,color,amsfonts,balance}
\usepackage{bm}
\usepackage[latin1]{inputenc}
\usepackage[update,prepend]{epstopdf}
\usepackage[noadjust]{cite}
\usepackage{mathtools}
\usepackage{multirow}
\usepackage{bbm} 
\usepackage{pdfpages}
\usepackage{tabulary}
\usepackage{multirow}
\usepackage{comment}
\usepackage{algorithm}
\usepackage{algorithmicx}
\usepackage{algpseudocode}
\usepackage{etoolbox}
\usepackage{enumerate}
\usepackage{caption}
\usepackage{subcaption}
\usepackage[svgnames]{xcolor}
\def\beq{\begin{equation}}
\def\eeq{\end{equation}}
\def\beqa{\begin{eqnarray}}
\def\eeqa{\end{eqnarray}}
\def\beqan{\begin{eqnarray*}}
\def\eeqan{\end{eqnarray*}}
\include{notation_new}

\usepackage{cite}
\newcommand\blfootnote[1]{%
  \begingroup
  \renewcommand\thefootnote{}\footnote{#1}%
  \addtocounter{footnote}{-1}%
  \endgroup
}

\newtheorem{Theorem}{Theorem}
\newtheorem{Proposition}{Proposition}

\newtheorem{Corollary}{Corollary}

\theoremstyle{definition}

\makeatletter
\def\munderbar#1{\underline{\sbox\tw@{$#1$}\dp\tw@\z@\box\tw@}}
\makeatother

\newcommand{\tran}{^{\text{\sf T}}}
\newcommand{\Abf}{\mathbf{A}}

\newcommand{\tbf}{\mathbf{t}}

\DeclarePairedDelimiter\floor{\lfloor}{\rfloor}

\begin{document}

\title{MIMO Networks with One-Bit ADCs: Receiver Design and Communication Strategies}
\author{Abbas Khalili$^1$, Farhad Shirani$^2$, Elza Erkip$^1$, Yonina C. Eldar$^3$\\
$^1$Dept. of Electrical and Computer Engineering,
New York University. \\
$^2$Dept. of Electrical and Computer Engineering,
North Dakota State University. \\
$^3$Dept. of Mathematics and Computer Science, Weizmann Institute of Science.}

\maketitle

\begin{abstract}
\blfootnote{This work is supported by NSF grant SpecEES-1824434 and NYU WIRELESS Industrial Affiliates. This work was presented in part at the IEEE International Symposium on Information Theory (ISIT), July 2019.}

High resolution analog to digital converters (ADCs) are conventionally used at the receiver terminals to store an accurate digital representation of the received signal, thereby allowing for reliable decoding of transmitted messages. However, in a wide range of applications, such as communication over millimeter wave and massive multiple-input multiple-output (MIMO) systems, the use of high resolution ADCs is not feasible due to power budget limitations. In the conventional fully digital receiver design, where each receiver antenna is connected to a distinct ADC, reducing the ADC resolution leads to performance loss in terms of achievable rates. One proposed method to mitigate the rate-loss is to use analog linear combiners leading to design of hybrid receivers. Here, the hybrid framework is augmented by the addition of delay elements to allow for temporal analog processing. Two new classes of receivers consisting of delay elements, analog linear combiners, and one-bit ADCs are proposed. The fundamental limits of communication in single and multi-user (uplink and downlink) MIMO systems employing the proposed receivers are investigated. In the high signal to noise ratio regime, it is shown that the proposed receivers achieve the maximum achievable rates among all receivers with the same number of one-bit ADCs.

\end{abstract}

\section{Introduction}

Millimeter wave (mmWave) and massive multiple-input multiple-output (MIMO) have emerged as promising technologies for high data rate communication in 5G wireless networks. These systems use large antenna arrays at both transmitter and receiver terminals to perform beamforming and spatial multiplexing. Nonetheless, application of large antenna arrays results in high energy consumption which is a major obstacle in the implementation of mmWave systems \cite{rangan2014millimeter,walden1999analog,Murmann2015}.

In conventional receiver design, each antenna is connected to a high resolution analog-to-digital converter (ADC) \cite{DigBF1}. Therefore, using a large number of antennas corresponds to using a large number of high resolution ADCs. The power consumption of an ADC is proportional to the number of quantization bins and hence grows exponentially in the number of output bits \cite{walden1999analog,Murmann2015}. The application of high resolution ADCs is especially costly when using channels with large bandwidths since the power consumption of an ADC grows linearly in bandwidth \cite{BR}. As a result, high resolution ADCs become a major source of receiver power consumption in mmWave and massive MIMO systems. One method proposed to address this high power consumption involves using few low resolution ADCs (e.g. one-bit threshold ADCs) \cite{abbas2019MT,abbas2019highsnr,Abbas2020Thr,mo2014channel,MIMO1,mezghani2012capacity,abbasISIT2018,rini2017generalITW,gong2020rf,ioushua2019family,9294107,shlezinger2019hardware,shlezinger2019asymptotic,shlezinger2021deep,alkhateeb2014mimo,Mo2017ADC,dabeer2006limits,mezghani2008analysis,koch2013low,dutta2019case}. Reducing the number of ADCs and their resolutions decreases power consumption, however, it also results in lower transmission rates. This suggests a tradeoff between transmission rate and power consumption which is controlled by the number and resolution of the ADCs at  the receiver.

 \begin{figure}[t]
\centering 
\includegraphics[width=0.6\textwidth,draft=false]{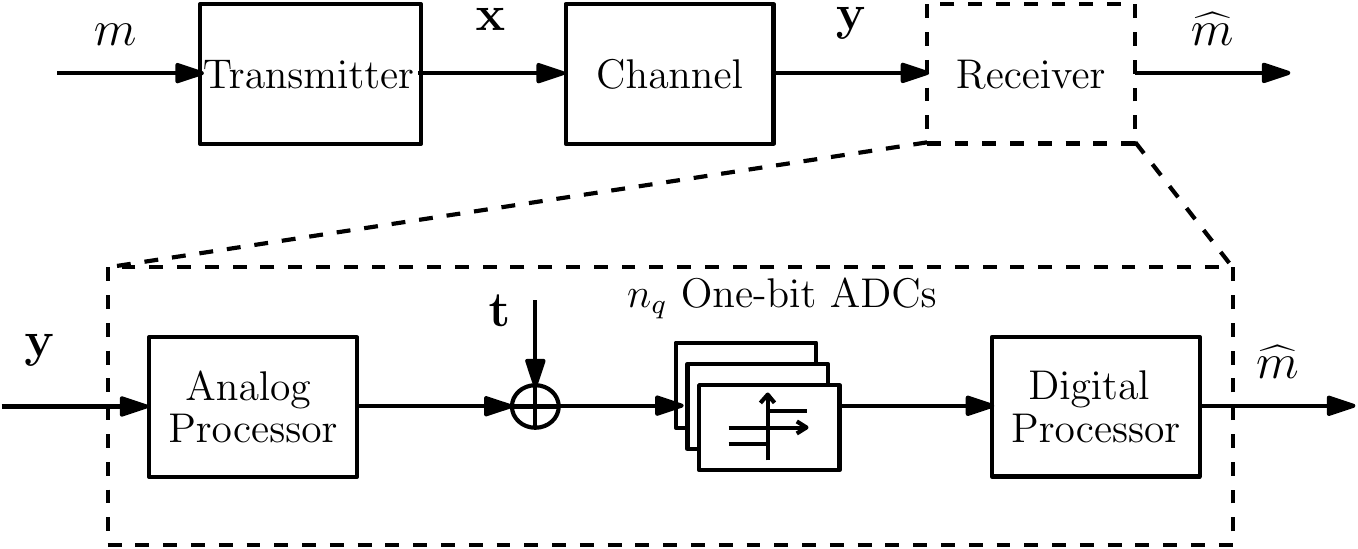}
\caption{A MIMO system with $n_t$ transmit antennas and $n_r$ receive antennas with a hybrid receiver consisting of an analog processor, $n_q$ one-bit ADCs with threshold vector $\mathbf{t}$, and a digital processor.}
\label{fig:classic}
\end{figure}

Consider the MIMO communication scenario shown in Fig. \ref{fig:classic}. For a receiver equipped with $n_q$ one-bit ADCs, the  channel  capacity  is  upper bounded  by  $\min\{n_q, C\}$  bits per channel-use, where $C$ is the capacity of the MIMO channel without any limitations on quantization. In this paper, we focus on linear processing in the analog domain and investigate how close we can get to this upper bound. Linear analog processing of one channel-output at a time (which we call \textit{hybrid one-shot}) is widely considered in the literature \cite{Mo2017ADC,molisch2017hybrid,abbasISIT2018,rini2017generalITW}. It has been shown that the maximum high SNR rate achievable using hybrid one-shot receivers is strictly less than $n_q$  \cite{abbasISIT2018}. By introducing analog temporal processing, we propose two new receivers namely \textit{hybrid blockwise} and \textit{adaptive threshold} receivers that overcome this limitation. In particular, we show that one can achieve $n_q$ bits per channel-use at high SNR. We note that temporal processing is also considered in \cite{shlezinger2019hardware,shlezinger2019asymptotic,shlezinger2021deep} but for quantization rather than channel coding.
 
There is a large body of work on communication schemes when receivers with low resolution ADCs are used. For communication over single-input single-output (SISO) channels, \cite{dabeer2006limits,mezghani2008analysis} show that the capacity achieving transmit signal has a discrete distribution, \cite{koch2013low} shows that ADCs with asymmetric thresholds achieve higher rates than ADCs with symmetric thresholds, and \cite{dutta2020quantization-arxiv} provides a linear model of the system studying the effect of low resolution quantization in linear (e.g., OFDM) systems.  References \cite{mezghani2007ultra,mo2015capacity} consider communication over single-user MIMO channels, where each antenna at the receiver is connected to a one-bit ADC per dimension (real and imaginary). In particular, \cite{mezghani2007ultra} shows that the achievable rate of the quantized MIMO system is approximately $2/\pi$ times less than the system with no quantization constraint at low SNRs. References \cite{alkhateeb2014mimo,Mo2017ADC,rini2017generalITW,abbasISIT2018}
consider receivers which perform linear analog preprocessing prior to quantization. For example, \cite{abbasISIT2018} shows, using geometric interpretations, that non-zero thresholds can improve the capacity of the system.

While much of the literature focuses on single-user communication, the use of low resolution ADCs at the receivers in multiterminal communications gives rise to new challenges in interference management and successive decoding schemes. In \cite{Rassouli2018}, communication over the multiple-access channel (MAC) is studied when each transmitter is equipped with a single antenna and the receiver has a single one-bit  ADC. It is shown that the optimal input distribution is discrete. Uplink communication over wireless networks, modeled as a MAC, is also studied in \cite{Studer2016,Hong2018ADC}, where practical transmission schemes are considered and analyzed when coarse quantization is used at the receiver. References \cite{liang2016mixed,Zhang2016ADC,Zhang2017ADC,Wen2015ADC,Jeon2018ADC,MOllen2017ADC} investigate the performance of massive MIMO systems with low resolution ADCs at the receiver. Finally, \cite{Wen2015ADC,Jeon2018ADC,MOllen2017ADC} consider mixed ADC receivers, where the receiver is connected to a set of low and high resolution ADCs. However, the capacity region of general single and multi-user communications in which the receivers are equipped with a fixed number of ADCs is not known. In particular, whether one can improve the performance by using a more sophisticated receiver design is not well understood.

The setting in this paper is communication over MIMO channels, where a fixed number of one-bit ADCs $n_q$ are used at the receiver. The conventional hybrid one-shot receiver achieves the high SNR rate of $n_q$ bits per channel-use only when $n_q$ is less than the rank of the channel and so this receiver does not achieve the full potential of the ADCs \cite{abbasISIT2018,mo2015capacity}. Our main contribution is to consider temporal processing in the analog domain before quantization to better utilize the ADCs, thereby increasing the communication rate.
Our contributions are summarized as follows:
\begin{itemize}

\item We propose two new MIMO receivers, hybrid blockwise and adaptive threshold, that increase the achievable rate (or the rate region in multi-user settings) compared with the hybrid one-shot receiver. These receivers perform temporal processing of the observed channel outputs in the analog domain. The hybrid blockwise receiver performs temporal processing using a delay network and linear analog combiner while the adaptive threshold receiver uses a delay network and adaptive ADC thresholds.

\item We characterize the high SNR capacity of the hybrid blockwise receiver for single and multi-user MIMO systems as a function of the blocklength (number of the channel-outputs processed jointly in the analog domain). We show analytically that for large blocklengths, this receiver achieves the optimal high SNR capacity/capacity region under some conditions on channel ranks. In addition, our numerical evaluations show that it is possible to get close to the optimal performance even for small blocklengths.

\item For the adaptive threshold receiver, we provide coding schemes for communication over single and multi-user MIMO systems and derive achievable rate regions for arbitrary SNRs. We show analytically that this receiver achieves the optimal high SNR capacity/capacity region irrespective of channel ranks. Through simulations for all SNRs, we observe that the proposed transmission schemes achieve near Shannon capacity in single-user communication and Shannon capacity with TDMA in BC scenarios.

\end{itemize}

The rest of the paper is organized as follows.  Section \ref{sec:System Model} describes the system model and required definitions. 
Sections \ref{sec:arch} and \ref{Sec:FBW} discuss the proposed receivers along with analysis of the achievable rates and rate regions.
Section \ref{sec:numerical_results} presents numerical analysis and simulation results. 
Section \ref{sec:conclusion} concludes the paper.  

Throughout the paper we use the following notation. Sets are denoted by calligraphic letters such as $\mathcal{X}, \mathcal{U}$. The set of natural and real numbers are denoted by $\mathbb{N}$ and $\mathbb{R}$, respectively. 
The set of numbers $\{1,2,\cdots, n\}, n\in \mathbb{N}$ is represented by $[n]$; $\mathbf{H}$ is a matrix, $\mathbf{h}$ is a vector, and $h$ and $H$ are scalars. The subvector $(x_k,x_{k+1},\cdots,x_n)$ is denoted by $\mathbf{x}_k^n$. We write $||\mathbf{x}||_2$ to denote the $L_2$-norm of $\mathbf{x}$ and $\mathbf{H}^{T}$ is the transpose of $\mathbf{H}$. 
The value of $i$ modulo $k$ is represented by $\mathrm{mod}_k(i), i,k\in \mathbb{N}$. The binary entropy function is $\mathrm{\mathrm{h_b}}(x) = -x\log{x} -(1-x)\log(1-x)$. 

\section{System Model and Preliminaries}
\label{sec:System Model}
In this section, we first provide the descriptions of the communication systems considered in this paper. Next, we explain the conventional hybrid one-shot receiver whose performance is used as a benchmark for the proposed receivers. Finally, we provide the required combinatorial background for our results.
\subsection{System Model}
In our derivations, we look at three communication scenarios described below.

\textbf{Single-user MIMO:} A MIMO system where the transmitter and receiver are equipped with $n_t$ and $n_r$ antennas, respectively, and the receiver has $n_q$ one-bit ADCs. Let  $\mathbf{H} \in \mathbb{R}^{n_r\times n_t}$ be the channel gain matrix. The channel input and output are related as 
\begin{equation}
\mathbf{y} = \mathbf{H} \mathbf{x} + \mathbf{n},
\label{eq:channel}
\end{equation}
where $\mathbf{x} \in \mathbb{R}^{n_{t}}$ with $E[|\mathbf{x}|^2] \leq P$ and $\mathbf{y}\in \mathbb{R}^{n_{r}}$ are the vectors representing transmit and received signals, receptively, and $\mathbf{n}$ is an independent noise vector of length $n_r$ consisting of independent and identically distributed (i.i.d.) unit variance and zero mean Gaussian random variables. 

\textbf{MIMO MAC:}
A $n_u$-user MIMO MAC where the receiver is equipped with $n_q$ one-bit ADCs and $n_r$ receive antennas. The transmitters have $n_{t,j}, j\in[n_u]$ transmit antennas. Let $\mathbf{H}_j\in \mathbb{R}^{n_r\times n_{t,j}}$ denote the channel gain matrix between transmitter $j$ and the receiver and $\mathbf{x}_j \in \mathbb{R}^{n_{t,j}}$ represent the channel inputs of user $j$. Then, the inputs and output of the channel are related through
 \begin{align}
\label{eq:MACMIMO} 
    \mathbf{y} = \sum_{j\in[n_u]}\mathbf{H}_j\mathbf{x}_j+\mathbf{n},    
 \end{align}
 where $\mathbf{y} \in \mathbb{R}^{n_{r}}$ is the vector of the received channel output, $E[|\mathbf{x}_j|^2]\leq P_j$, and $\mathbf{n} \in \mathbb{R}^{n_r}$ is an independent noise vector with i.i.d. unit variance and zero mean Gaussian entries.  

\textbf{MIMO BC:}
A $n_u$-user MIMO BC where receiver $j,~j\in [n_u]$ is equipped with $n_{r,j}$ antennas and employs $n_{q,j}$ one-bit ADCs. The transmitter has $n_t$ transmit antennas and the matrix $\mathbf{H}_j \in \mathbb{R}^{n_{r,j}\times n_t}$ represents the channel gain matrix between the transmitter and receiver $j$. Let $\mathbf{x} \in \mathbb{R}^{n_t}$ denote the channel input vector. The channel output at receiver $j$ is given by
  \begin{align}
     \mathbf{y}_j= \mathbf{H}_j\mathbf{x}+ \mathbf{n}_j,~ j\in [n_u],
 \end{align}
 where $\mathbf{y}_j \in \mathbb{R}^{n_{r,j}}$ is the channel output at receiver $j$, $E[|\mathbf{x}|^2]<P$, and $\mathbf{n}_j \in \mathbb{R}^{n_{r,j}}$ is the independent noise vector at receiver $j$ whose elements are i.i.d. Gaussian with unit variance and zero mean.

The channel gain matrices in all the above scenarios are assumed to be fixed during communication and known at both the transmitter and receiver.

We consider a receiver as in Fig.~\ref{fig:classic}, where the channel outputs are first processed in the analog domain. Then, the resulting analog signals are fed to $n_q$ one-bit ADCs with variable thresholds. Our goal is to maximize the rate region for the described communication scenarios under fixed number of antennas and one-bit ADCs. Towards this goal, we introduce two new receivers and optimize their parameters. For our proposed receivers and communication strategies, we build upon the conventional hybrid one-shot receiver. This receiver uses linear analog processing of the received signals before quantization and was proposed to mitigate the rate-loss due to coarse quantization of the channel outputs \cite{abbasISIT2018,rini2017generalITW,mezghani2012capacity}. 

\subsection{Hybrid One-shot Receiver}
\label{sec:AOS}
Consider the receiver depicted in Fig. \ref{fig:lin_comb}, where at each channel-use, an analog linear combiner multiplies the channel output vector $\mathbf{y}$ by a spatial analog processing matrix $\mathbf{V}$ and the resulting streams are input to one-bit ADCs with threshold vector $\mathbf{t}$. Denoting the ADC output with $\mathbf{w}$, we have
\begin{align}
\label{eq:wq}
    \mathbf{w} = {\rm Q}(\mathbf{V}\mathbf{y} + \mathbf{t}),
\end{align}
where ${\rm Q}(\cdot)$ is an element-wise function which returns the sign of each element.

The spatial analog processing matrix $\mathbf{V}$ and threshold vector $\mathbf{t}$ are assumed to be fixed over the transmission block which does not allow for temporal processing of the received signals in the analog domain. Hence, we name this receiver \textit{hybrid one-shot} receiver. The maximum achievable rate of hybrid one-shot receiver, maximized over all input distributions, threshold vectors, and linear combining matrices is denoted by $C_{HOS}$ and refers to the capacity of the hybrid one-shot receiver.

It was shown in prior works that hybrid one-shot communication over single-user MIMO systems with $n_q$ one-bit ADCs at the receiver inflicts a rate-loss due to which in many cases even at high SNR, the maximum achievable rate is strictly less than $n_q$ \cite{abbasISIT2018,mo2015capacity}. In this work, by introducing temporal analog processing and adaptive thresholds, we provide two new receivers which can achieve $n_q$ bits per channel-use at high SNR in scenarios where hybrid one-shot cannot.

\subsection{Combinatorial Background}

For the analysis of our proposed receivers, we utilize a geometric interpretation introduced in \cite{abbasISIT2018,mo2015capacity} which along with the needed combinatorial background, is briefly described next in the context of a single-user MIMO.

Loosely speaking, as SNR is increased, the effect of noise in \eqref{eq:channel} becomes negligible and the channel output becomes almost equal to $\mathbf{H}\mathbf{x}$. In fact, in the absence of noise, the channel output space is $\textrm{span}(\mathbf{H})$ (the subspace spanned by columns of the channel gain matrix $\mathbf{H}$).

Each one-bit ADC can be viewed as a hyperplane that partitions $\textrm{span}(\mathbf{H})$ into two. To elaborate, based on \eqref{eq:wq}, element $i$ of vector $\mathbf{w}$ for a received vector $\mathbf{y} \in \textrm{span}(\mathbf{H})$ is ${w}_i= {\rm Q}(\mathbf{v}_i^T\mathbf{y}+t_i)$, where $\mathbf{v}_i \in \mathbb{R}^{n_r}$ is row $i$ of the combiner matrix $\mathbf{V}$ and $t_i$ is element $i$ of the threshold vector $\mathbf{t}$. From a geometric perspective, the set of points $\mathbf{z} \in \mathbb{R}^{n_r}$ satisfying  $\mathbf{v}_i^T\mathbf{z}+t_i = 0$ is a hyperplane in $\mathbb{R}^{n_r}$ and if $\mathbf{v}_i^T\mathbf{y}+t_i > 0 ~(<0)$, it means that $\mathbf{y}$ is above (below) the hyperplane.

The hyperplanes  $\{\mathbf{v}_i^T\mathbf{z}+t_i = 0| i \in [n_q]\}$ partition $\textrm{span}(\mathbf{H})$ into set of decision regions that result in unique ADC output vectors ${\mathbf{w}}$.
The number of decision regions is the number of messages that can be reliably transmitted at high SNR. Even though there are $2^{n_q}$ possibilities for $\mathbf{w}$, based on the tuple $( \mathbf{H},\mathbf{V},\mathbf{t})$, the number of decision regions might be fewer since some of the output sign vectors are not realizable. For example, let $n_t=n_r=1$, $n_q=2$, $\mathbf{H}=1$,  $\mathbf{V}= (1,1)^T
$, and $\mathbf{t}= (0,0)^T$. Then, there are no points in $\textrm{span}(\mathbf{H}) = \mathbb{R}$ which result in ${\mathbf{w}} = (+1,-1)^T$ or ${\mathbf{w}} = (-1,+1)^T$. Therefore, in order to maximize the high SNR rate, we should choose $(\mathbf{V},\mathbf{t})$ such that the number of decision regions is maximized. Let us denote the set of decision regions in $\textrm{span}(\mathbf{H})$ for a given pair of $( \mathbf{V},\mathbf{t})$ by ${\cal{R}}(\mathbf{H}, \mathbf{V},\mathbf{t})$. The maximum number of decision regions is given by the following Proposition.

\begin{figure}[t]
\centering
\includegraphics[width=0.6\textwidth,draft=false]{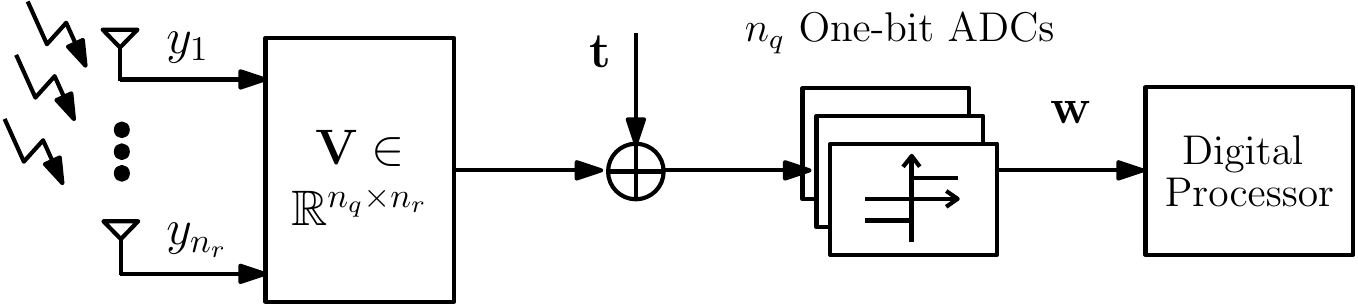}
\caption{A hybrid one-shot receiver, where
the analog linear combiner is characterized by the matrix $\mathbf{V} \in \mathbb{R}^{n_q\times n_r}$, and the ADC thresholds are $\mathbf{t}=(t_1,t_2,\cdots,t_{n_q})$.}
\label{fig:lin_comb}
\end{figure}

\begin{Proposition}(\hspace{-0.002in}\cite{winder1966partitions})
\label{Prop:Partition}
For the hybrid one-shot receiver in Fig.~\ref{fig:lin_comb} characterized with $(\mathbf{V}, \mathbf{t})$, for a channel gain matrix $\mathbf{H}$, the maximum number of decision regions is given by
\begin{align}
\label{eq:M_th}
    \max_{ \mathbf{V} \in \mathbb{R}^{n_q\times n_r},\mathbf{t} \in \mathbb{R}^{n_q}} |{\cal{R}}(\mathbf{H},\mathbf{V}, \mathbf{t})|
    = \sum_{i=0}^{\mathrm{rank}(\mathbf{H})} {n_q \choose i}.
\end{align}
Additionally, if the thresholds are set to zero,
\begin{align}
\label{eq:M_zth}
   \max_{ \mathbf{V} \in \mathbb{R}^{n_q\times n_r}}|{\cal{R}}(\mathbf{H},\mathbf{V},\mathbf{0})|
    = 2\sum_{i=0}^{\mathrm{rank}(\mathbf{H})-1} {n_q-1 \choose i}.
\end{align}
\end{Proposition}
Matrix $\mathbf{V}$ and threshold vector $\mathbf{t}$ that lead to the maximum number of decision regions as in Prop.~\ref{Prop:Partition} can be designed in different ways. For example, if the elements of $\mathbf{V}$ and $\mathbf{t}$ are i.i.d. using an arbitrary density function, with high probability, \eqref{eq:M_th} and \eqref{eq:M_zth} are satisfied. For more discussion on the required conditions on $\mathbf{V}$ and $\mathbf{t}$, we refer the reader to \cite{winder1966partitions}.

\section{ Hybrid Blockwise Receiver}
\label{sec:arch}

Shannon's channel capacity theorem suggests jointly decoding a large block of the channel outputs achieves the channel capacity.
Conventionally, joint processing of the channel outputs is done digitally since it requires non-linear arithmetics. In contrast, our proposed receivers perform temporal processing in analog before quantization. In this section, we discuss the \emph{hybrid blockwise} receiver, which utilizes delay elements and a linear analog combiner to perform analog spatial and temporal processing, thereby generalizing the hybrid one-shot receiver.

In the following, we provide an example of the high SNR rate-loss observed for the hybrid one-shot receiver in a single-user single-input single-output (SISO) scenario. Then, show how to mitigate this rate-loss by performing analog temporal processing using delay elements and a linear combiner. This will motivate our derivations in the rest of this section.
\subsection{Example: Single-User Hybrid Blockwise Receiver}
\label{subsec:ex_hyb}
\begin{figure}[t]
    \centering
    \includegraphics[width=0.8\textwidth,draft = false]{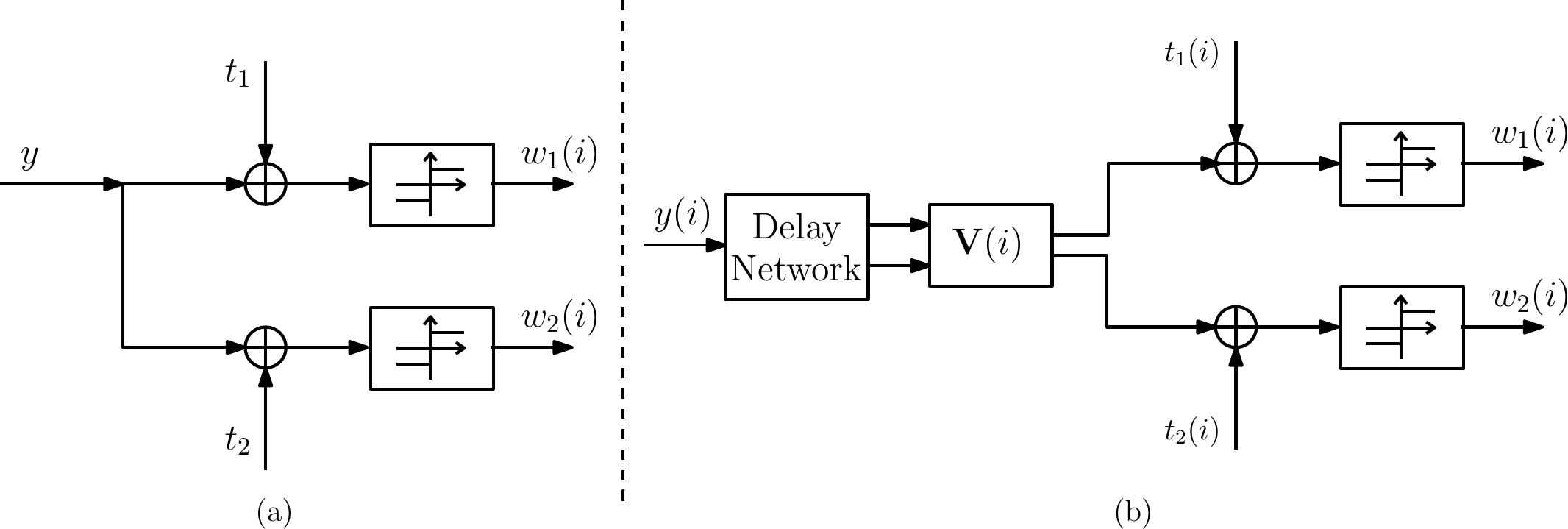}
    \caption{(a) A hybrid one-shot receiver with fixed thresholds. (b) A hybrid blockwise receiver with blocks of length two ($\ell  =2 $), where the receiver operates jointly on two channel uses. }
    \label{fig:exe1_rec}
\end{figure}

Consider a single-user SISO system (i.e. $n_t=n_r=1$), where the receiver is equipped with two one-bit ADCs (i.e. $n_q=2$). Fig. \ref{fig:exe1_rec}a shows a generic hybrid one-shot receiver whose decision regions are plotted in Fig.~\ref{fig:exe1_constel}. Given a channel output $y(i)$, where $i$ refers to the $i^{\rm th}$ channel-use, this receiver produces the digitized signal:
\begin{align}
\begin{aligned}
    ({w}_1(i),{w}_2(i))=
    \begin{cases}
    (-1,-1) \qquad & \text{if } y(i)<t_1,\\
    (+1,-1)&   \text{if } t_1<y(i)<t_2,\\
    (+1,+1)&   \text{if } t_2<y(i),
    \end{cases}
    \end{aligned}
\end{align}
    where we assumed $t_1<t_2$ without loss of generality. Note that the symbol $(-1,+1)$ is not produced at this receiver. At high SNR, the communication noise is negligible and rates close to $\log{3} = 1.585$ bits per channel-use are achievable which is the high SNR capacity of this receiver \cite{abbasISIT2018}.

Next, we show how using the hybrid blockwise receiver shown in Fig. \ref{fig:exe1_rec}b increases the high SNR capacity. Here, the receiver jointly process every two consecutive channel outputs over two channel uses. To elaborate, consider the linear analog combiner and threshold vector pairs 
\begin{align}
\label{eq:vandt}
(\mathbf{V}_{\mathrm{even}}, \mathbf{t}_{\mathrm{even}})\!=\!\! \left(\begin{bmatrix}
         1 &0\\
         0 &1
    \end{bmatrix}\!,\!\begin{bmatrix}
         0.5 \\
         0.5 
    \end{bmatrix}\right)\!,~~
    (\mathbf{V}_{\mathrm{odd}}, \mathbf{t}_{\mathrm{odd}})\!=\!\! \left(\begin{bmatrix}
         \cos(\pi/4) & \sin(\pi/4)\\
         \cos(3\pi/4) & \sin(3\pi/4)
    \end{bmatrix}\!,\!\begin{bmatrix}
         -0.5 \\
         -0.5 
    \end{bmatrix} \right)\!,
\end{align}
and let us consider the first four channel uses. In the first two channel uses, the receiver is idle. In the third and fourth channel uses, the receiver processes $y(1)$ and $y(2)$ using the linear analog combiner and threshold vector pairs $(\mathbf{V}_{\mathrm{odd}},\tbf_{\mathrm{odd}})$ and  $(\mathbf{V}_{\mathrm{even}},\tbf_{\mathrm{even}})$ from \eqref{eq:vandt}, respectively. By the fifth channel-use the receiver has observed two new channel outputs (i.e., $y(3)$ and $y(4)$) and repeats the process. 
As a result, the number of decision regions as shown in Fig.~\ref{fig:exe1_constel} increases compared to the hybrid one-shot receiver.
At high SNR, since the noise is negligible, this receiver achieves the maximum rate of $\frac{\log{11}}{2} = 1.7297$ bits per channel-use (the factor $2$ is in the denominator because it takes two channel uses to determine each decision region) which is greater than the high SNR rate of the hybrid one-shot receiver with two one-bit ADCs.

\begin{figure}[t]
    \centering
    \includegraphics[width=0.4\textwidth,draft = false]{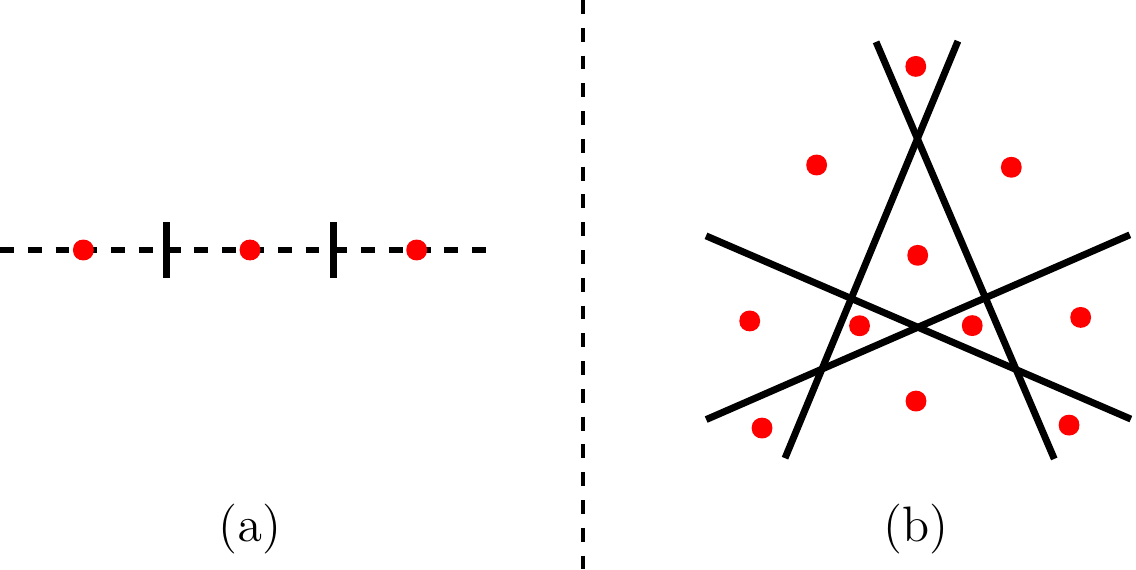}
    \caption{Decision regions for the (a) hybrid one-shot receiver shown in Fig.~\ref{fig:exe1_rec}a and (b) hybrid blockwise receivers shown in Fig.~\ref{fig:exe1_rec}b, where red dots represent possible transmitted constellations.}
    \label{fig:exe1_constel}
\end{figure} 

\begin{figure}[b!]
\centering
\includegraphics[width=0.7\textwidth,draft=false]{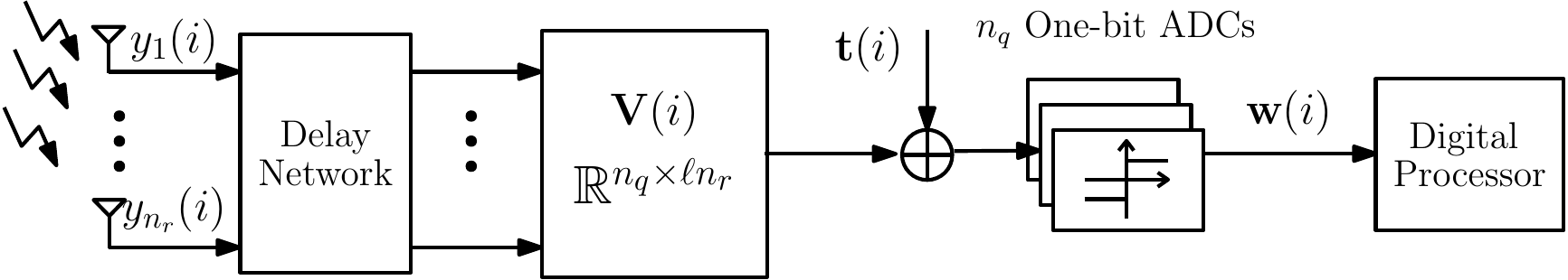}
\caption{A hybrid blockwise receiver, where
the linear combiner at the $i^{\rm th}$ channel-use is characterized by  $\mathbf{V}(i)$ and threshold vector by $\mathbf{t}(i)=(t_1(i),t_2(i),\cdots,t_{n_q}(i))$.}
\label{fig:BW_Rx}
\end{figure}

\subsection{Proposed Receiver}
\label{sec:prop_hyb}
 We now provide a formal description of the proposed hybrid blockwise receiver illustrated in Fig. \ref{fig:BW_Rx}. In this receiver, delay elements and a fixed set of linear analog combiners and thresholds are used to perform linear temporal and spatial processing of the channel outputs before quantization. The receiver processes every $\ell \in \mathbb{N}$ consecutive channel outputs together. Towards this goal, a delay network consisting of $2\ell$ delay elements shown in Fig.~\ref{fig:delay net} is used. Each delay element $D(\cdot)$ is of size $n_r$ (number of receiver antennas) and delays its input vector by one channel-use. Therefore, using this delay network, we can store and access a channel output $\mathbf{y}$ for $2\ell$ channel uses. 
 
 The receiver has $\ell$ pairs of analog combiner matrix and threshold vector using which it processes $\ell$ channel outputs jointly over $\ell$ channel uses. To elaborate, consider the first $2\ell$ channel uses. During channel uses $i\leq \ell$, the delay network is empty and the receiver is idle. During channel uses $ \ell < i \leq 2\ell$, the channel outputs $\mathbf{y}(i), i\in[\ell]$ are available in the delay network and the receiver jointly process them with the $\ell$ linear combiner and threshold vector pairs using one pair at each channel-use. In general, during channel uses $(m-1)\ell< i\leq m\ell$, the receiver jointly process channel outputs $\mathbf{y}(i), (m-2)\ell  <i\leq (m-1)\ell$ using the $\ell$ pairs. The described operation is shown in Table \ref{tb:Arch1}. Note that since during the first $\ell$ channel uses the receiver is idle, we will have a constant delay of $\ell$ channel uses at the start of the transmission for the hybrid blockwise receiver. In Sec.~\ref{sec:numerical_results}, we will show that small values of $\ell$ are sufficient to achieve near optimal high SNR rate.

\begin{table}[b!]
\centering
\caption{Operation timeline of a hybrid blockwie receiver with delay of length $\ell$.}
\label{tb:Arch1}
\includegraphics[width=0.8\textwidth,draft=false]{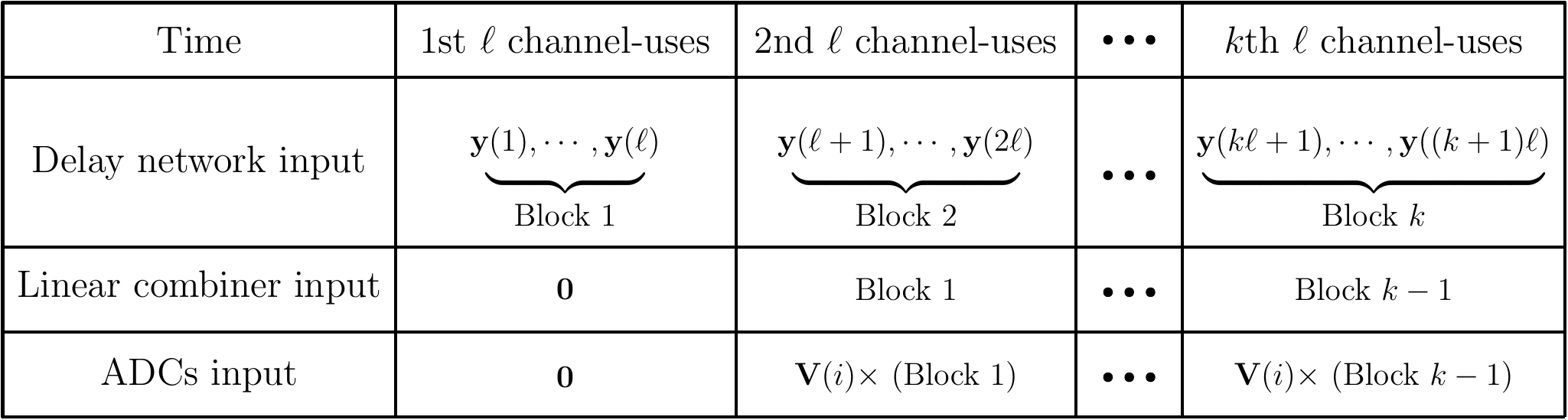}
\end{table}

\begin{figure}[t!]
    \centering
     \includegraphics[width=0.5\textwidth ,draft=false]{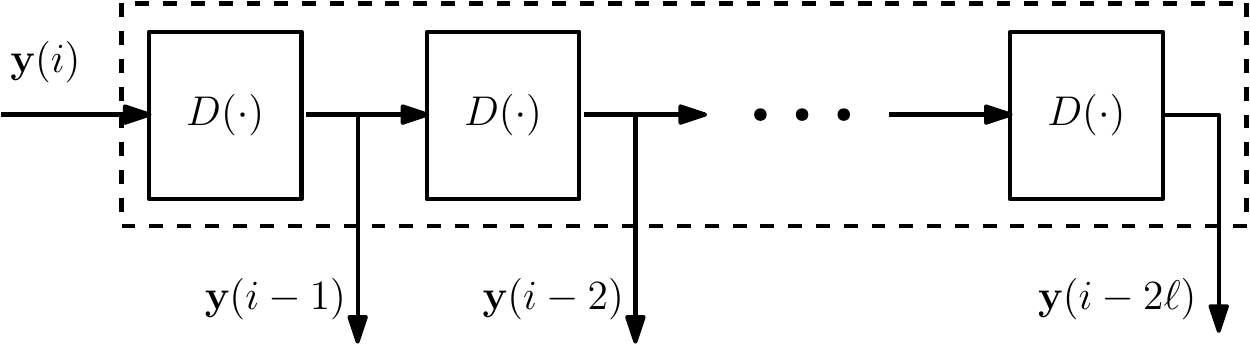}
    \caption{Delay network in the $i^{\rm th}$ channel-use.}
    \label{fig:delay net}
\end{figure}
The maximum achievable rate for a given delay $\ell$ optimized over all analog combining matrices, and threshold vectors is denoted by $C_{\ell}$ and corresponds to the capacity of the hybrid blockwise receiver. 

Note that the hybrid one-shot receiver illustrated in Fig. \ref{fig:lin_comb} is a special case of the hybrid blockwise receiver, where the analog combiner at the receiver only processes one channel output at a time (i.e. $\ell=1$).  
As a result, 
\begin{align}
C_{HOS}=C_{1}\leq C_{\ell},~\ell > 1.
\end{align}
\subsection{Achievable High SNR Single-User Rate}
\label{subsec:PtP_hyb}
Consider the single-user MIMO communication system described in Section \ref{sec:System Model} with the hybrid blockwise receiver employing $n_q$ one-bit ADCs. The pair $(\mathbf{V},\mathbf{t})$ describing the blockwise receiver are taken so that the number of decision regions is maximized as described in Prop.~\ref{Prop:Partition}. In the next theorem, we characterize the high SNR capacity and perform asymptomatic analysis as the number of delay elements $\ell$ is increased.

\begin{Theorem}
\label{th:PtP}

For the single-user MIMO communication system described in Section \ref{sec:System Model} with $n_t$ transmit and $n_r$ receive antennas and channel gain matrix $\mathbf{H}$ where the hybrid blockwise receiver with $n_q$ one-bit ADCs is used, the high SNR capacity $C^{\mathrm{High~SNR}}_{\ell}$ satisfies:
\begin{align}
    \label{eq:rate_ptp_hyb}
\frac{1}{\ell} \log \left( \sum_{i = 1}^{\ell \mathrm{rank}(\mathbf{H})} {\ell n_q \choose i}\right) \leq C^{\mathrm{High~SNR}}_{\ell} \leq \frac{1}{\ell} \log \left( \sum_{i = 1}^{\ell n_r}{\ell n_q \choose i}\right).
\end{align}
When the ADCs threshold vector is set to zero ($\mathbf{t} = \mathbf{0}$), we have
\begin{align}
\frac{1}{\ell} \log \left( 2\sum_{i = 1}^{\ell \mathrm{rank}(\mathbf{H})-1} {\ell n_q-1 \choose i}\right) \leq C^{\mathrm{High~SNR}}_{\ell} \leq \frac{1}{\ell} \log \left( 2\sum_{i = 1}^{\ell n_r-1}{\ell n_q-1 \choose i}\right).
\end{align}
Furthermore, for both zero and optimized threshold vectors, high SNR capacity satisfies:
\begin{itemize}
    \item For a fixed $n_q$ and large $\ell$,
\begin{align}
\label{eq:linf}
\begin{aligned}
n_q \mathrm{h_b}(\alpha)-&\frac{1}{2\ell}\log{\ell}
 + O\left(\frac{1}{\ell}\right)
\leq C^{\mathrm{High~SNR}}_{\ell} \leq 
n_q \mathrm{h_b}(\beta)-\frac{1}{2\ell}\log{\ell}
+O\left(\frac{1}{\ell}\right),\\
&\alpha= \min \left\{\frac{\mathrm{rank}(\mathbf{H})}{n_q},\frac{1}{2}\right\},~\beta=\min\left \{\frac{n_r}{n_q},\frac{1}{2}\right\}.
\end{aligned}
\end{align}
\item For a fixed $\ell$ and large $n_q$,
\begin{align}
\label{eq:nqinf}
    \mathrm{rank}(\mathbf{H}) \log{n_q} + O(1)   \leq C^{\mathrm{High~SNR}}_{\ell}  \leq       n_r\log{n_q} + O(1).
\end{align}
\end{itemize}

\end{Theorem}
\begin{proof}
The proof of above theorem is provided in Appendix~\ref{app:th1}.
\end{proof}
As a result of this theorem, we have the following corollary:
\begin{Corollary}
\label{cor:ptp_hyb}
The hybrid blockwise receiver achieves the high SNR rate of $n_q$ bits per channel-use in single-user MIMO when $n_q\leq 2\mathrm{rank}(\mathbf{H})$ for large $\ell$. 
\end{Corollary}

Next, we provide communication strategies and their corresponding high SNR achievable rates for the MIMO MAC (uplink) and MIMO BC (downlink) communication systems employing hybrid blockwise receiver.

\subsection{Achievable High SNR MAC Rate}
\label{subsec:MAC_hyb}
For the MIMO MAC scenario with $n_q$ one-bit ADCs at the receiver, we consider a time-sharing method which we argue achieves the optimal high SNR sum-rate. Assume that user $j$, where $j\in[n_u]$, transmits for $\eta_j,~0<\eta_j<1$ portion of the total channel uses with $\sum_{j\in[n_u]}\eta_j = 1$. If each user employs the communication strategy developed for the single-user case (Thm.~\ref{th:PtP}), at high SNR, user $j$ achieves $\eta_j  C^{\mathrm{High~SNR}}_{\ell}$ bits per channel-use which leads us to following result. 
\begin{Corollary}
\label{cor:MAC_hyb}
The hybrid blockwise receiver achieves the high SNR sum-rate of $n_q$ bits per channel-use in MIMO MAC when $n_q \leq 2\displaystyle\min_{j\in[n_u]} \mathrm{rank}(\mathbf{H}_{j})$.
\end{Corollary}

For the case of MIMO BC, we need a more sophisticated form of time sharing and processing. The  following  example  describes  the  key ideas we use to establish the high SNR capacity bounds in Sec.~\ref{subsec:BC_hyb}.
\subsection{Example: BC with Hybrid Blockwise Receiver}
\label{Ex:BC_hyb}

\begin{figure}[t]
 \centering
\includegraphics[width=0.5 \textwidth,draft=false]{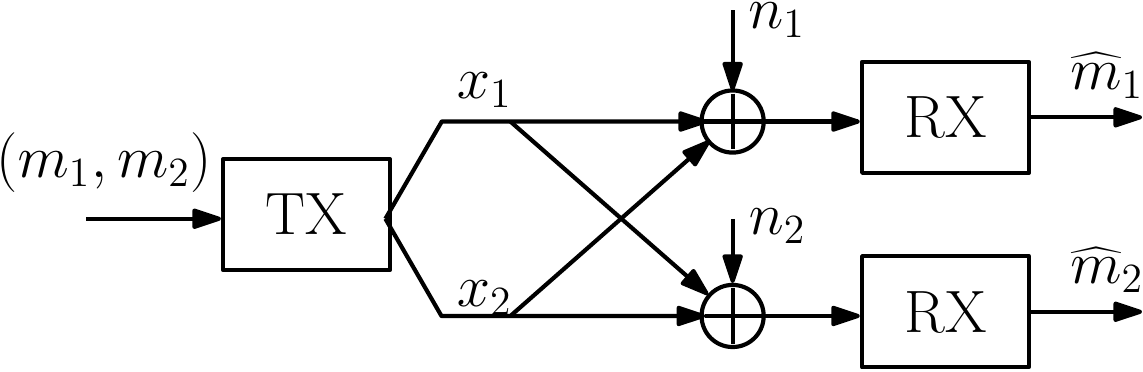}
\caption{A two user BC example. The transmitter is equipped with two antennas and each receiver is equipped with a single antenna and a one-bit ADC.}
\label{fig:BC}
\end{figure}

Consider the two-user BC channel shown in Fig. \ref{fig:BC}, where $n_t=2,n_{r,1}=n_{r,2}$, $n_{q,1}=n_{q,2}=1$, and $\mathbf{H}_1=\mathbf{H}_2=
(1, 1)$.

First consider using hybrid one-shot receiver and assume we use equal time-shares between the users. Since each receiver generates one bit for each observed channel output, both users achieve the high SNR rate of $0.5$ bits per channel use.

Next, we show how using hybrid blockwise receiver can improve this rate. The main idea is to use the ADCs to extract information from previously observed channel outputs when there are no new observations. Suppose we transmit to each user for two channel uses before switching to the other one. Each user employs a hybrid blockwise receiver that jointly process two channel observations over four channel uses. 

To elaborate, consider the linear analog combiner and threshold vector pairs 
\begin{align}
\label{eq:vandtBC}
\begin{aligned}
&(\mathbf{V}_{a}, \mathbf{t}_{a})\!=\!\! \left(\begin{pmatrix}
         1 &0
    \end{pmatrix}\!,\!
         0.5  
   \right)\!,~~
    (\mathbf{V}_{b}, \mathbf{t}_{b})\!=\!\! \left(\begin{pmatrix}
         \cos(\pi/4) & \sin(\pi/4)
    \end{pmatrix}\!,\!
         -0.5 
     \right)\!,\\
&    (\mathbf{V}_{c}, \mathbf{t}_{c})\!=\!\! \left(\begin{pmatrix}
         0 &1
    \end{pmatrix}\!,\!
         0.5 
    \right)\!,~~
    (\mathbf{V}_{d}, \mathbf{t}_{d})\!=\!\! \left(\begin{pmatrix}
         \cos(3\pi/4) & \sin(3\pi/4)
    \end{pmatrix}\!,\!
         -0.5 
     \right)\!,
    \end{aligned}
\end{align}
and let us consider the first six channel uses at the first user. In the first two channel uses, the receiver is idle. In the third, fourth, fifth, and sixth channel-use, the receiver processes $y(1)$ and $y(2)$ using the linear analog combiner and threshold vector pair $(\mathbf{V}_{a},\tbf_{a})$, $(\mathbf{V}_{b},\tbf_{b})$, $(\mathbf{V}_{c},\tbf_{c})$, and $(\mathbf{V}_{d},\tbf_{d})$ from \eqref{eq:vandtBC}, respectively. By the seventh channel-use, the receiver has observed two new channel outputs and repeats the process. The decision regions of this receiver are shown in Fig.~\ref{fig:exe1_constel}b.
At high SNR, since the noise is negligible, this receiver achieves the maximum rate of $\frac{\log{11}}{4} = 0.8648$ bits per channel-use (the factor $4$ is in the denominator because it takes four channel uses to determine each decision region) which is greater than the high SNR rate of the hybrid one-shot receiver with equal time sharing.

\subsection{Achievable High SNR BC Rate}
\label{subsec:BC_hyb}
Generalizing the example in Sec.~\ref{Ex:BC_hyb} to MIMO BC scenario, we consider a time sharing strategy in which we transmit to user $j$, where $j\in[n_u]$ in $\eta_j \ell$ channel uses with $0<\eta_j<1$. Each user employs the hybrid blockwise receiver of Sec.~\ref{sec:prop_hyb}. However, user $j$ process every $\eta_j\ell$ channel observations during $\ell$ channel uses with $\ell$ different pairs of linear analog combiners and thresholds $(\mathbf{V}, \mathbf{t})$. This leads to the following theorem.
\begin{Theorem}
\label{th:BC_hyb}
Based on the above communication scheme, the maximum achievable high SNR rate $R^{\mathrm{High~SNR}}_{j,\ell}$ of user $j$, where $j\in[n_u]$, satisfies:
\begin{align}
    \label{eq:rate_BC_hyb}
\frac{1}{\ell} \log \left( \sum_{i = 1}^{\eta_j\ell \mathrm{rank}(\mathbf{H}_j)} {\ell n_{q,j} \choose i}\right) \leq R^{\mathrm{High~SNR}}_{j,\ell} \leq \frac{1}{\ell} \log \left( \sum_{i = 1}^{\ell \eta_jn_r}{\ell n_{q,j}\choose i}\right).
\end{align}
When the ADCs threshold vector is set to zero ($\mathbf{t} = \mathbf{0}$), we have
\begin{align}
\frac{1}{\ell} \log \left( 2\sum_{i = 1}^{\eta_j\ell \mathrm{rank}(\mathbf{H}_j)-1} {{\ell n_{q,j}}-1 \choose i}\right) \leq R^{\mathrm{High~SNR}}_{j,\ell}\leq \frac{1}{\ell} \log \left( 2\sum_{i = 1}^{\eta_j{\ell n_r}-1}{{\ell n_{q,j}}-1 \choose i}\right).
\end{align}
Furthermore, for both zero and optimized threshold vectors, the maximum high SNR achievable rate satisfies:
\begin{itemize}
    \item For a fixed $n_{q,j}$ and large $\ell$,
\begin{align}
\label{eq:linf_BC}
\begin{aligned}
 n_{q,j}\mathrm{h_b}(\alpha_j)-&\frac{1}{2\ell}\log{\ell}
 + O\left(\frac{1}{\ell}\right)
\leq R^{\mathrm{High~SNR}}_{j,\ell} \leq 
n_{q,j} \mathrm{h_b}(\beta_j)-\frac{1}{2\ell}\log{\ell}
+O\left(\frac{1}{\ell}\right),\\
&\alpha_j= \min \left\{\frac{\eta_j\mathrm{rank}(\mathbf{H}_{j})}{n_{q,j}},\frac{1}{2}\right\},~\beta_j=\min\left \{\frac{\eta_jn_r}{n_{q,j}},\frac{1}{2}\right\}.
\end{aligned}
\end{align}
\item For a fixed $\ell$ and large $n_q$,
\begin{align}
\label{eq:nqinf_BC}
    \eta_j\mathrm{rank}(\mathbf{H}_j) \log{n_{q,j}} + O(1)   \leq R^{\mathrm{High~SNR}}_{j,\ell} \leq  \eta_j n_r\log{n_{q,j}} + O(1).
\end{align}
\end{itemize}
\end{Theorem}
We then have the following corollary:
\begin{Corollary}
\label{cor:BC_hyb}
The hybrid blockwise receiver achieves the optimal high SNR rate of $n_{q,j}$ bits per channel-use for user $j$ in MIMO BC when $n_{q,j} \leq 2\eta_j\mathrm{rank}(\mathbf{H}_{j})$, where $(\eta_j)_{j\in[n_u]}$ is such that $0<\eta_j<1$ and $\sum_{j\in [n_u]}\eta_j = 1$.
\end{Corollary}
\subsection{Observations}
The following observations are based on Thm.~\ref{th:PtP}, and \ref{th:BC_hyb}, and Cor.~\ref{cor:ptp_hyb},~\ref{cor:MAC_hyb}, and \ref{cor:BC_hyb} for the hybrid blockwise receiver.
\\\noindent {\bf I)} The single-user high SNR capacity $C^{\mathrm{High~SNR}}_{\ell}$ increases with $\ell$ and $n_q$. Furthermore, for large $\ell$ or $n_q$, $C^{\mathrm{High~SNR}}_{\ell}$ for the cases of zero and non-zero threshold ADCs converge. This shows that when long delays $\ell$ can be tolerated or the number of ADCs is large,  zero threshold ADCs which are easier to implement can be used without any loss in high SNR rate.
\\\noindent {\bf II)} The hybrid blockwise receiver achieves the optimal high SNR rate of $n_q$ bits per channel-use for the single-user MIMO when $n_q\leq 2\mathrm{rank}(\mathbf{H})$. This improves the results of the hybrid one-shot receiver (e.g. \cite{abbasISIT2018}, \cite{mo2015capacity}), where the high SNR rate of $n_q$ bits per channel-use is only achieved when $n_q \leq \mathrm{rank}(\mathbf{H})$. Note that $n_q \leq \mathrm{rank}(\mathbf{H})$ effectively means less than a one-bit ADC per antenna which could be limiting. 
\\\noindent {\bf III)} The hybrid blockwise receiver achieves the optimal high SNR sum-rate of $n_q$ bits per channel-use in MIMO MAC and the optimal per user high SNR rates of of $n_{q,j}$ bits per channel-use in MIMO BC under conditions on the channel ranks.

 \section{Adaptive Threshold Receiver}
\label{Sec:FBW}
The optimality of the hybrid blockwise receiver at high SNR can only be established under conditions on the rank of the channel, as stated in Cor.~\ref{cor:ptp_hyb},~\ref{cor:MAC_hyb}, and \ref{cor:BC_hyb}. In this section, we propose a more complex receiver which incorporates adaptively changing the ADC thresholds at each channel-use based on their outputs in the previous channel uses. We show that the proposed receiver achieves the optimal high SNR rate in single-user MIMO, per user high SNR rate in MIMO BC, and high SNR sum-rate in MIMO MAC irrespective of the channel ranks. 

Before providing a formal description of the proposed receiver, we consider the following example that motivates the use of adaptive thresholds.

\subsection{Example: Single-User Adaptive Threshold Receiver}
\label{Ex:2}

\begin{figure}[t]
 \centering
\includegraphics[width=0.7\textwidth, draft=false]{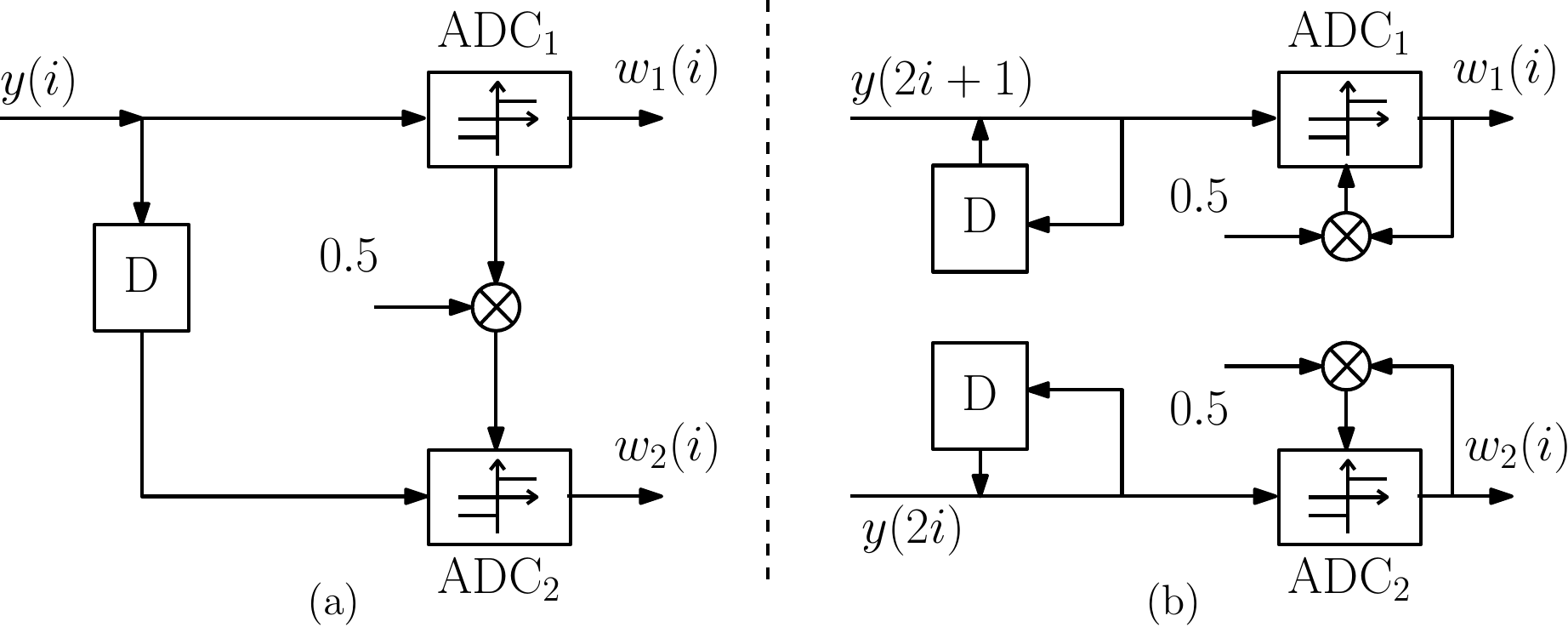}
\caption{Two adaptive threshold receivers with two one-bit ADCs that achieve the optimal rate of two bits per channel-use at high SNR for a single-user SISO system.}
\label{fig:simple}
\end{figure}

Figure~\ref{fig:simple}a shows an example of the adaptive threshold receiver with two one-bit ADCs for a single-user SISO channel. Let us consider the first two channel uses. In the first channel-use, the received signal is input to ADC$_1$ whose threshold is zero. In the second channel-use, the same signal is input to ADC$_2$ whose threshold is set to half of ADC$_1$ output from the first channel-use. Meanwhile, ADC$_1$ quantizes the received signal in the second channel-use. Considering that the receiver operates in a similar fashion for the rest of the communication, we have:
\begin{align*} ({w}_1(i),{w}_2(i+1))=
    \begin{cases}
    (-1,-1) \qquad & \text{if } y(i)<-\frac{1}{2},\\
    (-1,+1)&   \text{if } -\frac{1}{2}<y(i)<0,\\
    (+1,-1)&   \text{if } 0<y(i)<\frac{1}{2},\\
    (+1,+1)&   \text{if } \frac{1}{2}<y(i).
    \end{cases}
\end{align*}
As a result, at high SNR for $i>1$, we get two bits per channel-use and for $i=1$, we get $1$ bit since ADC$_2$ is idle. Hence, the rate of two bits per channel-use is achieved. Note that this rate cannot be improved since there are only two one-bit ADCs.

Figure~\ref{fig:simple}a is not the only way of utilizing adaptive thresholds that leads to the optimal high SNR rate. For example, consider the receiver in Fig.~\ref{fig:simple}b.
Here, received signals at odd and even channel uses are input to ADC$_1$ and ADC$_2$, respectively. Each ADC processes its input signal over two channel uses with variable thresholds. In the first channel-use, the threshold is zero and in the second channel-use, it is set to half of the ADC's previous output. Similar to Fig.~\ref{fig:simple}a, this receiver also achieves the optimal rate of two bits per channel-use at high SNR.

Extending the receivers in Fig.~\ref{fig:simple}, we can argue that for a SISO channel using $n_q$ one-bit ADCs with adaptive thresholds, $n_q$ bits per channel-use at high SNR can be achieved. The main idea is to perform $n_q$-bit quantization of each received signal over $n_q$ channel uses. To elaborate, in the first channel-use processing a signal, we use an ADC with zero threshold, and in the $i^{\rm th}$ channel-use $i> 1$ of processing it, we use the threshold $\sum_{j = 1}^{i-1} 2^{n_q-j-1}b_j$, where $b_j \in \{-1,+1\}$ is the $j^{\rm th}$ bit extracted from the signal. To implement this, one can use a receiver as in Fig.~\ref{fig:simple}a, where there is a fixed temporal order among the ADCs and the threshold of each ADC is a linear function of the previous ADCs' outputs. Another possible implementation is as in Fig.~\ref{fig:simple}b, where an ADC is used for $n_q$ channel uses and its threshold is updated based on its previous outputs. Since each ADC extracts one-bit at each channel-use and the presence of $n_q$ ADCs lets us process $n_q$ signals in parallel, we can hence achieve $n_q$ bits per channel-use at high SNR. Note that in the first $n_q-1$ channel-uses, the rate is less than $n_q$ bits per channel-use as some of the ADCs are idle. However, aside from this constant delay at the start of the transmission, this receiver does not cause any other delays.

There are already different types of ADCs with adaptive thresholds in the literature such as successive approximation register and sigma-delta ADCs which have been proposed for reducing power consumption or area of the circuit \cite{5746277,5711005,5433830,6043594}. Therefore, one might be able to utilize these ADCs the same way as the ADCs in Fig.~\ref{fig:simple} to reduce the power consumption.

\begin{figure}[t]
\centering
\includegraphics[width= 0.7\textwidth, draft = false]{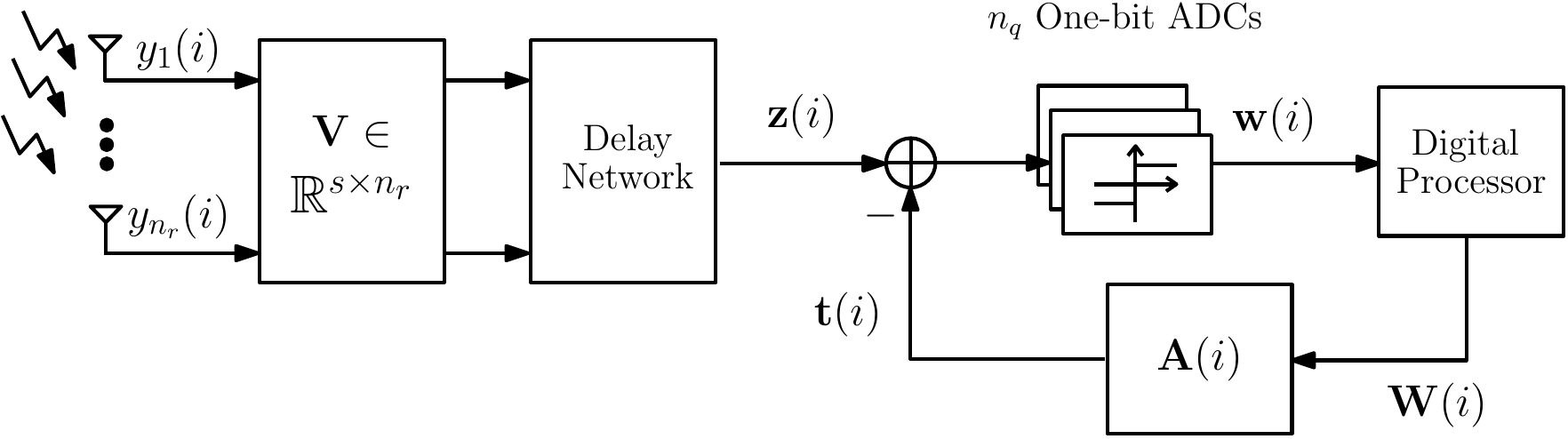}
\caption{ An adaptive threshold receiver using a linear analog combiner with $s$ output streams, $n_q$ one-bit ADCs, delay network, and adaptive thresholds.}
\label{fig:PtP_2}
\end{figure}

\subsection{Proposed Receiver}
\label{sec:prp_adap}

The block diagram of the proposed adaptive threshold receiver for a MIMO system is presented in Fig. \ref{fig:PtP_2}. A linear analog combiner is used to perform one-shot (i.e., spatial) processing of the received signals similar to the hybrid one-shot receiver. However, the combined streams are fed to a delay network that inputs them to a set of one-bit ADCs with adaptive thresholds.

In this receiver, we use a delay network similar to the one in Fig.~\ref{fig:delay net}. However, the size of each delay element is set to $s$ (the number of output streams of the analog combiner). To determine the length of the delay network, assume that the $k^{\rm th}$ output stream of the analog combiner, where $k\in[s]$, is allocated $n_{q,k}$ one-bit ADCs. We perfom $n_{q,k}$-bit quantization of each of the received signals in this stream in a manner similar to the example in Sec.~\ref{Ex:2}. A signal is processed over $n_{q,k}$ channel uses and one bit is extracted from it at each channel-use. Therefore, we require at least $n_{q,k}-1$ delay elements for the $k^{\rm th}$ stream. As a result, in total, it suffices to use a delay network of length $\displaystyle\max_{k\in[s]}{\{n_{q,k}\}}-1$. There are also other designs that can work. For example, the receivers in Fig.~\ref{fig:simple}a and Fig.~\ref{fig:simple}b have different delay networks and both achieve optimal high SNR rates.

At high SNR, it takes $n_{q,k}-1$ channel uses for the $k^{\rm th}$ stream to achieve the rate of $n_{q,k}$ bits per channel-use based on our discussion in the example in Sec.~\ref{Ex:2}. Therefore, this receiver has a constant delay of $\displaystyle\max_{k\in[s]}{\{n_{q,k}\}}-1$ channel uses before achieving the optimal high SNR rate of $n_q$ bits per channel-use. Furthermore, our analysis in \cite{Abbas2020Thr}, where we have provided practical mmWave simulations of this receiver, shows that having $n_{q,k}=4$ for all streams can lead to a near-optimal performance. This corresponds to a constant delay of only three channel uses at the start of communication.

To obtain ADC thresholds at each channel-use, we use a linear combination of the ADC outputs from previous channel uses similar to the example in Sec.~\ref{Ex:2}. For a stream with $n_{q,k}$ one-bit ADCs, it is enough to only consider ADC outputs from previous $n_{q,k}-1$ channel uses. Consider matrix ${\mathbf{W}}(i)$ whose columns are the ADC outputs from previous $\displaystyle\max_{k\in[s]}{\{n_{q,k}\}}-1$ channel uses, we can express the threshold vector at the $i^{\rm th}$ channel-use as
\begin{align}
    \tbf(i) = \mathbf{A}(i) {\rm vec}\left(\mathbf{W}(i)\right),
\end{align}
where ${\rm vec}(\cdot)$ is an operator whose output is a column vector resulting from concatenation of the columns of its input matrix. We refer to $\Abf(i)$ as the \textit{threshold coefficient} matrix that denotes the linear mapping for the thresholds.

As an example, for the receiver in Fig.~\ref{fig:simple}a, $s = 1$, ${\rm vec}\left(\mathbf{W}(i)\right) = [w_1(i-1)$,$w_2(i-1)$,$w_1(i)$, $w_2(i)]\tran$, and threshold coefficients matrix is $\mathbf{A} = [\mathbf{a}_1\tran$, $\mathbf{a}_2\tran]\tran$ with $\mathbf{a}_1(i) = (0,0,0,0)\tran$ and $\mathbf{a}_2(i) = (0.5,0,0,0)$. In this example, the threshold coefficient matrix $\Abf$ does not change with time. However, using a time-varying matrix $\Abf(i)$ becomes necessary in MIMO BC as it helps us extract more information from a signal over multiple channel uses. We will elaborate on this in Sec.~\ref{subsec:AT-Rx MT}.

Note that to adjust the thresholds, digital to analog converters (DACs) might be required. There are multiple ways that one can utilize DACs here. For example, one can use one-bit DACs for each element of ${\rm vec}\left(\mathbf{W}(i)\right)$ and then do the linear combination in analog, or do the linear combination in digital and use DACs with higher resolutions. Finding the best architecture from a power consumption perspective could be an interesting venue for future research. Here, we focus on the optimal utilization of one-bit ADCs to maximize the spectral efficiency.

\subsection{Achievable Single-User Rates}
\label{subsec:PtP_adap}

As shown in Fig.~\ref{fig:PtP_2}, the proposed adaptive threshold receiver first creates $s$ streams through linear spatial processing of channel outputs and then, processes these streams using a delay network and a set of one-bit ADCs with adaptive thresholds. This receiver allows us to effectively distribute the ADCs' resolutions among the streams. For example, given $s = 2$ and $n_q$ one-bit ADCs, we can effectively perform $n_{q,1}$ and $n_{q,2}$-bit quantization of the first and second stream, respectively, where $n_{q,1}+n_{q,2} = n_{q}$. Motivated by this observation, we provide a communication scheme for the single-user MIMO system described in Section \ref{sec:System Model}.

We consider singular value decomposition (SVD) of the channel, which results in a precoder at the transmitter and corresponding analog linear combiner at the receiver. As a result, the channel is transformed into $s$ parallel, non-interfering streams, where $s\leq {\rm rank}(\mathbf{H})$. Assuming that $n_{q,k}$ one-bit ADCs are allocated to the $k^{\rm th}$ stream, where $\sum_{k\in[s]} n_{q,k} = n_q$, the transmitter uses $2^{n_{q,k}}$-PAM with uniform distribution of the symbols to transmit over the stream. The receiver selects the threshold coefficients to recover the modulation points. The following theorem formulates the maximum achievable rate of this communication scheme. 

\begin{Theorem}
\label{th:PtP_adap}
Based on the above communication scheme, the achievable rate $R$, satisfies:
\begin{align}
\label{eq:rate_ptp_adp}
    R\leq \max \sum_{k =1}^{s} I(\widetilde{x}_k;\widetilde{y}_k),
\end{align}
where $s = \mathrm{rank}(\mathbf{H})$ is the number of streams, the maximum is taken over $(n_{q,k})_{k\in [s]}, (P_{k})_{k\in [s]}$ such that $\sum_{k\in [s]}n_{q,k} = n_q, \sum_{k\in[s]}P_{k}=P$, $\widetilde{x}_k = a_k\cdot\left(2\widehat{x}_k-1-2^{n_{q,k}}\right), \widetilde{y}_k = \sigma_{k} \widetilde{x}_k+\tilde{n}_k$, $a_k = \sqrt{\frac{3P_k}{2^{2n_{q,k}}-1}}$, $\widehat{x}_k$ is uniformly distributed over $[2^{n_{q,k}}]$, $\tilde{n}_k$s are i.i.d.  zero-mean Gaussian random variables with unit variance, and $\sigma_k$ is the $k^{\rm th}$ singular value of $\mathbf{H}$.
\end{Theorem}
\begin{proof}
The proof of the above theorem follows as in \cite{tse2005fundamentals}.
\end{proof}

As a result of this theorem, we have the following corollary:

\begin{Corollary}
\label{cor:ptp_adap}
The high SNR capacity of the adaptive threshold receiver is $n_q$ bits per channel-use. Therefore, the adaptive threshold receiver is optimal at high SNR.
\end{Corollary}

Note that due to practical constraints such as limited resolution of phase shifters, the accurate realization of SVD might not be feasible. This would cause power leakage among the streams and reduce the total rate of the system. To find a suitable substitute for the SVD matrix, one can use the analysis in \cite{Mo2017ADC}. The study of this effect is out of the scope of this paper and is left for future work. Also, one can show that finding the optimal power and ADC allocation in Thm.~\ref{th:PtP_adap} is an NP-hard problem \cite{bixby2004mixed}. We have investigated possible low complexity solutions to this problem in \cite{Abbas2020Thr}.

One advantage of the adaptive threshold receiver compared to hybrid blockwise receiver is that its high SNR capacity is linear in the number of ADCs irrespective of the channel rank. Recall that for the hybrid blockwise receiver, the high SNR capacity increases linearly only if $n_q\leq 2{\rm rank}(\mathbf{H})$. 

Next, we propose communication strategies and associated achievable rates for the MIMO MAC (uplink) and MIMO BC (downlink) communication systems employing adaptive threshold receivers.

\subsection{Achievable MAC Rate Region}
\label{subsec:AT-Rx MT}

Here, we use the same time-sharing scheme as in Sec.~\ref{subsec:MAC_hyb}. Assume that user $j$, where $j\in[n_u]$ transmits for $\eta_j,~0<\eta_j<1$ portion of the total channel uses with $\sum_{j\in[n_u]}\eta_j = 1$. If each user employs the communication strategy developed for the single-user case (Thm.~\ref{th:PtP_adap}), user $j$ achieves $\eta_j  R$ bits per channel-use, where $R$ is as in \eqref{eq:rate_ptp_adp}, which leads us to following result.
\begin{Corollary}
\label{cor:MAC_adap}
The adaptive threshold receiver achieves the high SNR sum-rate of $n_q$ bits per channel-use in MIMO MAC.
\end{Corollary}
For the case of MIMO BC, the following example describes the key  ideas we use to establish an achievable rate region in Sec.~\ref{sec:BC_adap}.  

\subsection{Example: BC with Adaptive Threshold Receiver}
\label{Ex:3}

Consider the two-user BC channel shown in Fig. \ref{fig:BC}, where $n_t=2,n_{r,1}=n_{r,2}=n_{q,1}=n_{q,2}=1$, and $\mathbf{H}_1=\mathbf{H}_2=
(1, 1)$.

Similar to the hybrid blockwise receiver example in  Sec.~\ref{Ex:BC_hyb}, the main idea is to use the one-bit ADCs to continue extracting information from previously observed channel outputs when there are no new observations. Consider the adaptive threshold receiver in Fig.~\ref{fig:simple}b. Let us use ADC$_1$ for the first user and ADC$_2$ for the second user. If we transmit to the first user in odd channel uses and to the second user in the even channel uses, they each can generate two bits per observed channel output over two channel uses and achieve the optimal high SNR rate of one bit per channel-use.

In general, over a blocklength of $\ell$, we can transmit to the first and second user over $\eta \ell$ and $(1-\eta)\ell$ channel uses, respectively. Then, we can utilize adaptive thresholds to perfom $\floor*{{1}/{\eta}}$ (or $\floor*{{1}/{(1-\eta)}}$ )-bit quantization of each of the observed channel outputs during $\floor*{{1}/{\eta}}$ (or $\floor*{{1}/{(1-\eta)}}$) channel uses at the first (second) receiver as discussed in the example in Sec.~\ref{Ex:2}. This way both users still achieve the optimal high SNR rate of one bit per channel-use.

The same idea can be applied for more than two users with more one-bit ADCs to achieve optimal high SNR rates per user. For example, when there are $n_u$ users equipped with $n_q$ one-bit ADCs, we can do time-sharing among the users, where over each block of length $\ell$ we transmit to user $j \in[n_u]$ for $\eta_j\ell$ channel uses. And at user $j$, we perform $\floor*{{n_q}/{\eta_j}}$ bit quantization of each channel output observation over $\floor*{{n_q}/{\eta_j}}$ channel uses. To elaborate, let us consider all users are equipped with an adaptive threshold receiver similar to that in Fig.~\ref{fig:simple}b, where each ADC processes a signal for $\floor*{{n_q}/{\eta_j}}$ channel uses to extract $\floor*{{n_q}/{\eta_j}}$ bits as discussed in Sec.~\ref{Ex:2}. Since there are $n_q$ ADCs at the receivers, a receiver can process $n_q$ channel outputs in parallel and hence achieve the optimal rate of ${\eta_j}\floor*{{n_q}/{\eta_j}}$ bits per channel use at high SNR.

\subsection{Achievable BC Rate Region}
\label{sec:BC_adap}

Generalizing the example in Sec.~\ref{Ex:3} to MIMO BC with arbitrary SNR, we use a time-sharing scheme, where over a blocklength of $\ell$, we transmit for $\eta_j \ell$ channel uses to receiver $j$, where $j\in[n_u]$, $0<\eta_{j}<1$, and $\sum_{j\in[n_u]}\eta_j = 1$. Each user employs the adaptive threshold receiver of Sec.~\ref{sec:arch} and a communication strategy similar to that of Thm.~\ref{th:PtP_adap}. However, receiver $j$ process every $\eta_j\ell$ channel observations during $\ell$ channel uses. To elaborate, consider the first user, we use SVD to divide its channel into $s_1$ parallel streams and allocate $n_{q,1,k}$ one-bit ADCs to the $k^{\rm th}$ stream, where $\sum_{k\in[s_1]} n_{q,1,k} = n_{q,1}$. Then, we perfom $\floor*{n_{q,1,k}/\eta_1}$-bit quantization over the $k^{\rm th}$ stream by processing each received signal during $\floor*{n_{q,1,k}/\eta_1}$ channel uses similar to the example in Sec.~\ref{Ex:3}. Therefore, we can decode  $2^{\floor*{n_{q,1,k}/\eta_1}}$-PAM over stream $k$. This leads to the following theorem.
\begin{Theorem}
\label{th:BC}
Based on the above communication scheme, the achievable rate $R_{j}$ of user $j$, where $j\in[n_u]$, satisfies:
\begin{align}
\label{eq:BCR}
\begin{aligned}
   & R_j \leq \eta_j
   \max \sum_{k \in[s_j]} I(\widetilde{x}_{j,k}, \widetilde{y}_{j,k}),
\end{aligned}
\end{align}
where $s_j, j\in [n_u]$ is the number of singular values of $\mathbf{H}_j$, the maximum is taken over $(n_{q,j,k})_{k\in [s_j]}$, $(P_{j,k})_{k\in [s_j]}$ such that $\sum_{k\in [s]}n_{q,j,k} = n_{q,j}, \sum_{k\in[s_j]}P_{j,k}=P_j$, $\widetilde{x}_{j,k} = a_{j,k}\left(2\widehat{x}_{j,k}-1-2^{\floor*{{n_{q,j,k}}/{\eta_j}}}\right)$, $\widetilde{y}_{j,k} = \sigma_{j,k} \widetilde{x}_{j,k}+n_{j,k}, k\in [s_j], j\in \{1,2\}$, $a_{j,k} = \sqrt{\frac{3P_{j,k}}{2^{2\floor*{{n_{q,j,k}}/{\eta_j}}}-1}}$, $\widehat{x}_{j,k}$ is uniformly distributed over $\left[2^{\floor*{{n_{q,j,k}}/{\eta_j}}}\right]$, and $\sigma_{j,k}$ is the $k^{\rm th}$ singular value of $\mathbf{H}_j$. 
\end{Theorem}

We then have the following corollary:
\begin{Corollary}
\label{cor:BC_adap}
The adaptive threshold receiver achieves the high SNR rate of $n_{q,j}$ bits per channel-use for user $j, j\in[n_u]$ in MIMO BC, respectively.
\end{Corollary}
\subsection{Observations}
The following observations are based on Thm.~\ref{th:PtP_adap}, and \ref{th:BC} and Cor.~\ref{cor:ptp_adap}, \ref{cor:MAC_adap}, and \ref{cor:BC_adap}:
\\\noindent {\bf I)} The single-user high SNR capacity of the adaptive threshold receiver is $n_q$. This improves the high SNR achievable rate compared to hybrid one-shot and hybrid blockwise receivers that only achieve $n_q$ for $n_q \leq \mathrm{rank}(\mathbf{H})$ and $ n_q \leq 2\mathrm{rank}(\mathbf{H})$, respectively.
\\\noindent {\bf II)} The adaptive threshold receiver achieves optimal high SNR sum-rate and per user rate in MIMO MAC and BC, respectively, irrespective of the channel ranks. This improves the high SNR achievable rate region compared to the hybrid blockwise receiver which requires conditions on the channel ranks (Cor.~\ref{cor:MAC_hyb}, and \ref{cor:BC_hyb}) to achieve the same performance.
\section{Numerical Evaluations and Simulation Results}
\label{sec:numerical_results}

In this section, we provide numerical analysis and simulation results for the proposed receivers and transmission schemes in Sections \ref{sec:arch} and \ref{Sec:FBW}. We first provide high SNR numerical analysis of the achievable rates of the hybrid blockwise and adaptive threshold receiver. Then, we provide simulation of the achievable rates of the adaptive threshold receiver in single-user and multi-user communication scenarios.

\subsection{High SNR Achievable Rate}

\begin{figure*}[t]
\centering
\begin{subfigure}{0.45\linewidth}
    \centering
    \includegraphics[width=\textwidth,draft=false]{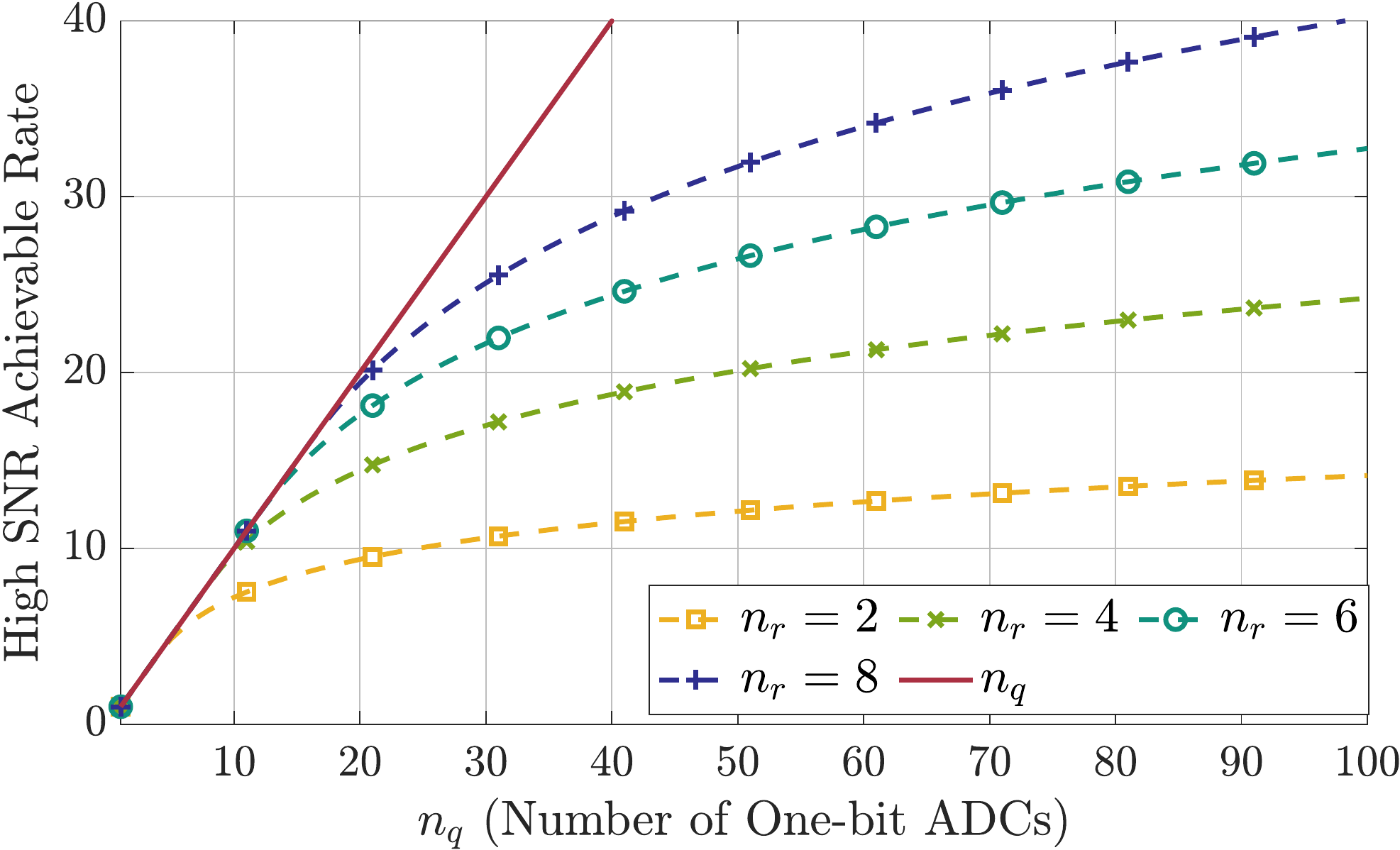}
    \caption{}
    \label{fig:infinite Delay}
\end{subfigure}
\begin{subfigure}{0.45\linewidth}
    \centering
 \includegraphics[width=\textwidth,draft=false]{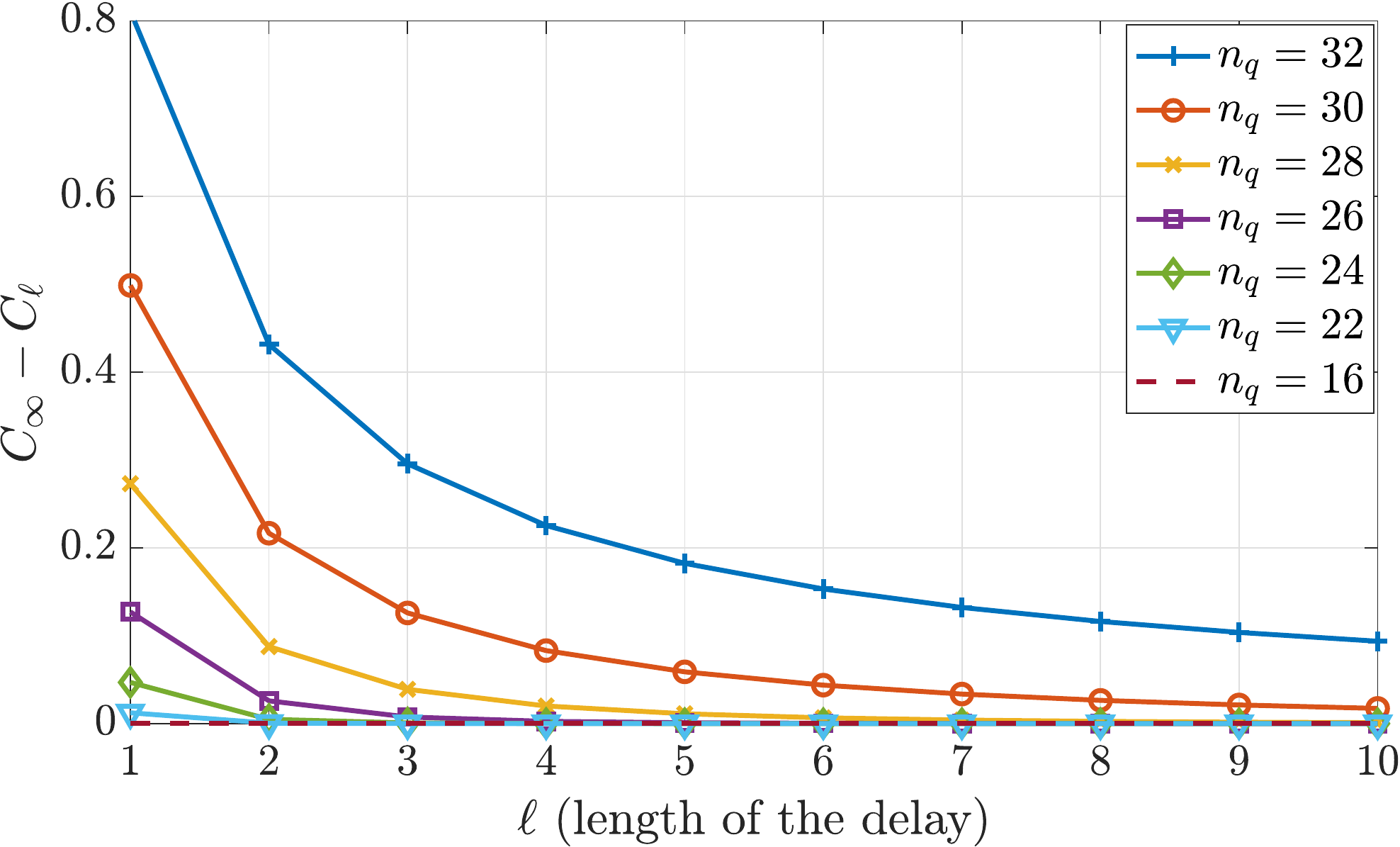}
    \caption{}
    \label{fig:finite delay}
\end{subfigure}
\caption{(a) Achievable high SNR rate of the hybrid blockwise receiver for large $\ell$ (dotted lines) for a MIMO system with $n_t=10$ and $n_r\in \{2,4,6,8\}$. The solid line, $R=n_q$ is the capacity upper bound achievable by adaptive threshold receiver. (b) High SNR rate-loss of the hybrid blockwise receiver when $\ell$ is finite assuming $\mathrm{rank}(\mathbf{H}) = 16$.}
\end{figure*}

\begin{figure}[t]
\centering
\includegraphics[width=0.45\textwidth,draft=false]{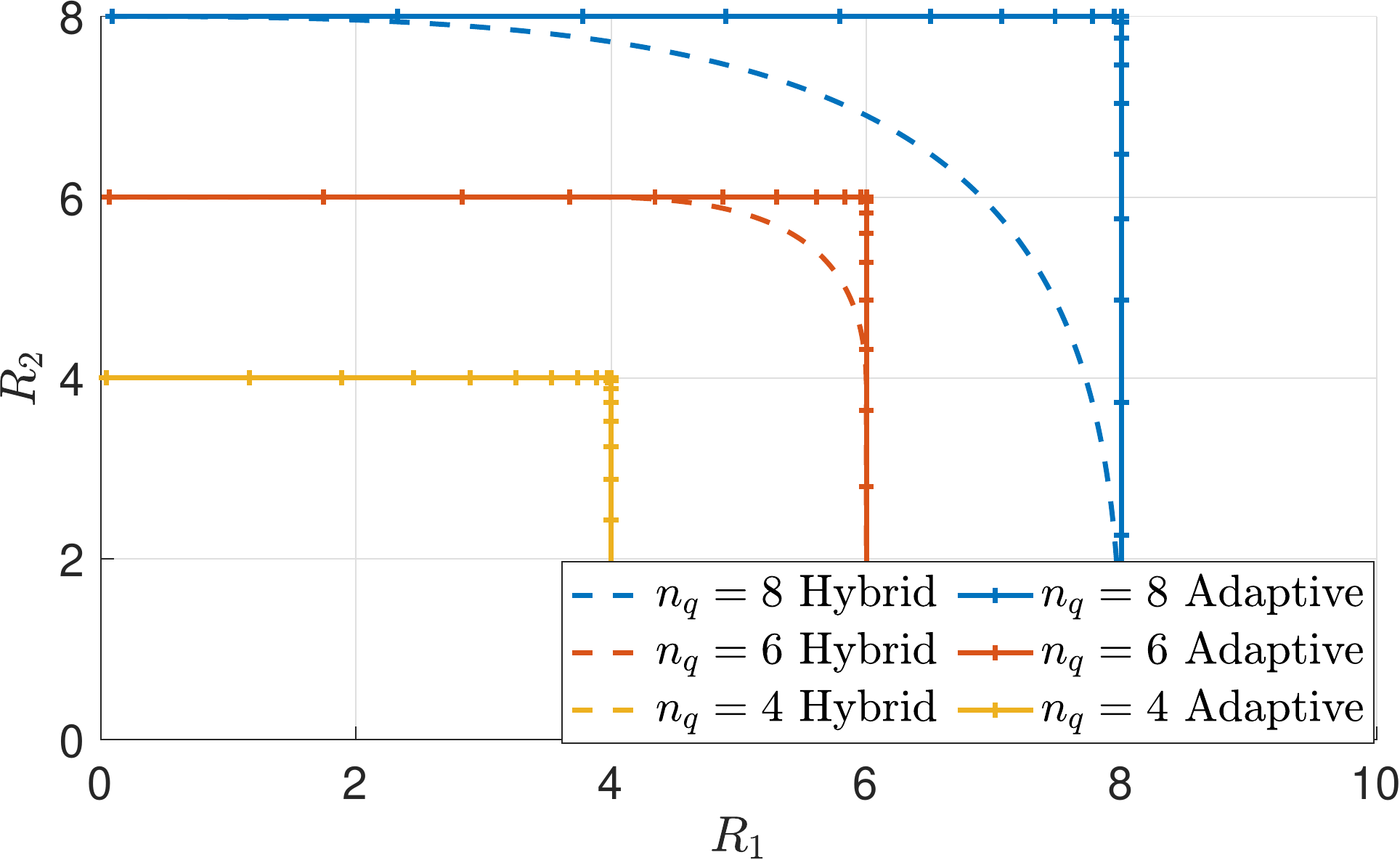}
\caption{Achievable high SNR rate region for the hybrid blockwise receiver (dashed lines) and adaptive threshold receiver (solid lines) for two-user MIMO BC employing $n_t = 10$ and identical receivers with $n_r = 4$, $n_q\in\{4,6,8\}$.}
\label{fig:BC_highSNR}
\end{figure}

We first illustrate Thm.~\ref{th:PtP}, that is high SNR achievable rates of the proposed hybrid blockwise receiver for single-user MIMO. The achievable rates (doted lines) are plotted in Fig. \ref{fig:infinite Delay} for large $\ell$ and different numbers of received antennas and quantizers assuming that the channel is full rank (i.e., $\mathrm{rank}(\mathbf{H}) =\min\{n_t, n_r\} $). Capacity upper bound $n_q$ is also shown for comparison. Note this upper bound is achievable when adaptive threshold receiver is used (Thm.~\ref{th:PtP_adap}). We observe that for a fixed number of transmit $n_t$ and receive $n_r$ antennas, as the number of one-bit ADCs $n_q$ is increased, the maximum high SNR rate with the adaptive threshold receiver always grows linearly with $n_q$. However, with the hybrid blockwise receiver the maximum high SNR rate increases linearly if $n_q\leq 2 \mathrm{rank}(\mathbf{H})$, logarithmically otherwise.

Next, we investigate the effect of delay $\ell$ on the high SNR achievable rate of the hybrid blockwise receiver (Thm.~\ref{th:PtP}). Fig. \ref{fig:finite delay}, illustrates the difference between the achievable high SNR rate of the hybrid blockwise receiver for large $\ell$ and when $\ell$ is finite. For $n_q\leq 16$, we have $n_q \leq {\rm rank}(\mathbf{H})$ which means that hybrid one-shot receiver ($\ell = 0$) would achieve the high SNR capacity of $n_q$ bits per channel use and so no temporal processing is needed.  
Furthermore, we observe that for $16<n_q \leq 32$, a delay length of $\ell =10$ is enough to get an achievable high SNR rate within $0.1$ bits per channel-use of the asymptotic case. As a result, using a finite value of $\ell$ is enough to achieve performance close to the infinite case.

Figure \ref{fig:BC_highSNR} shows the achievable high SNR rate region for two-user MIMO BC with hybrid blockwise (for large $\ell$) and adaptive threshold receivers. We observe that the high SNR rate region is a square for the adaptive threshold receiver as this receiver can always achieve the optimal high SNR rate per users. However, the hybrid blockwise receiver rate region can be smaller as it requires certain conditions on the channel ranks (Cor.~\ref{cor:MAC_hyb}, and \ref{cor:BC_hyb}) to achieve the same performance.

\begin{figure*}[t]
\centering
\begin{subfigure}{0.45\linewidth}
    \centering
    \includegraphics[width=\textwidth,draft=false]{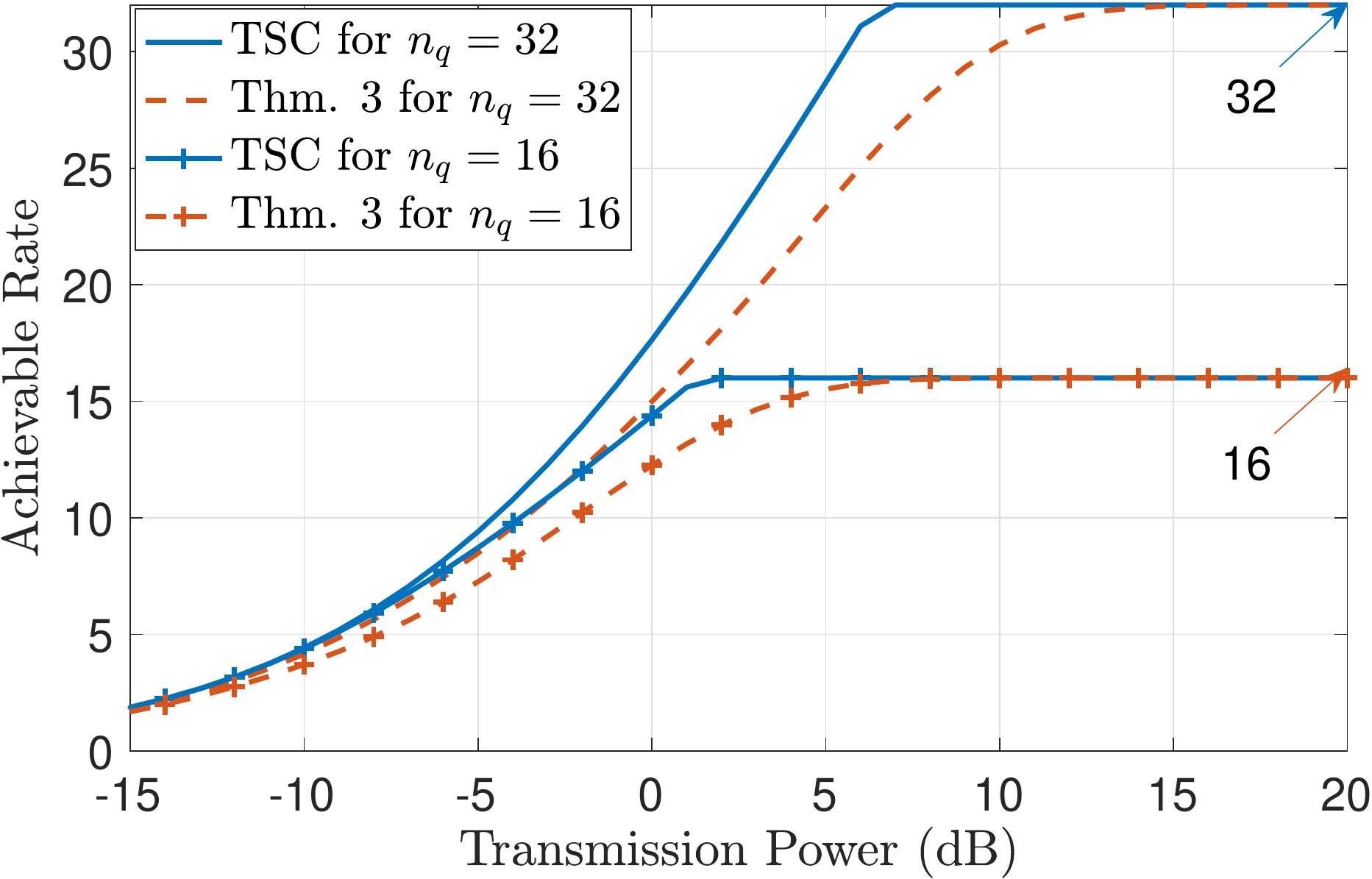}
    \caption{}
\label{fig:PtP_ns}
\end{subfigure}
\begin{subfigure}{0.45\linewidth}
    \centering
\includegraphics[width=\textwidth,draft=false]{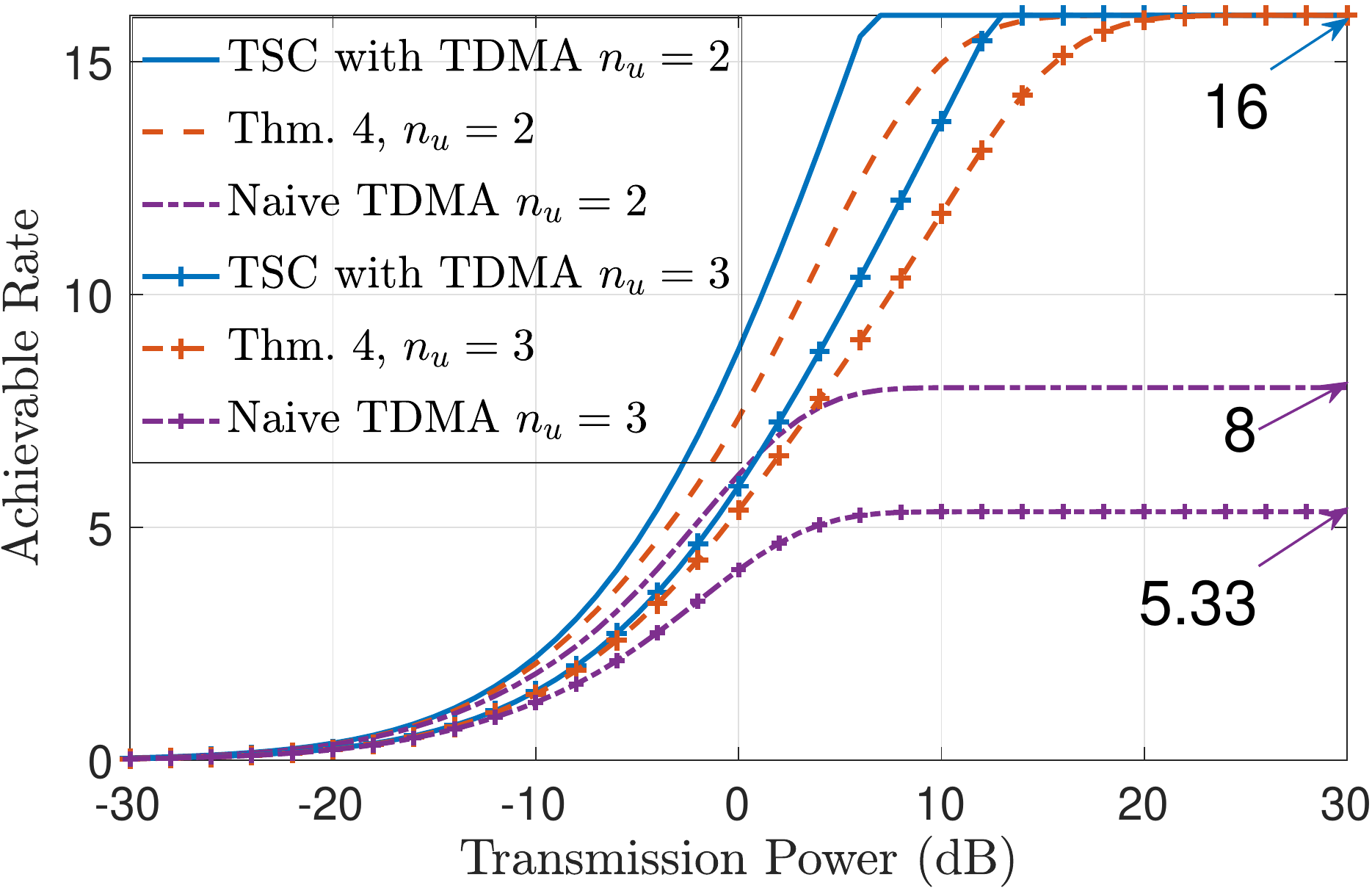}
    \caption{}
\label{fig:BC_sim2}
\end{subfigure}
\caption{(a) Achievable rate of a single-user MIMO system with the adaptive threshold receiver useing different number of ADCs $n_q = \{32,16\}$. TSC stands for truncated Shannon capacity. (b) Achievable rate per user for $n_u = \{2,3\}$ user BC with adaptive threshold receivers employing $n_q = 16$ one-bit ADCs.}
\end{figure*}
\subsection{Simulation Results for the Adaptive Threshold Receiver}

In this subsection, we study achievable rate regions for the adaptive threshold receiver (Section \ref{Sec:FBW}). We consider the single-user and BC MIMO scenarios. We assume a Rayleigh fading channel, where each element of the channel matrix is Gaussian with zero mean and variance one and average the simulated rate regions over $10000$ channel realizations. For simulation of a practical mmWave system, we refer the reader to \cite{Abbas2020Thr}.

The proposed transmission schemes in Thm.~\ref{th:PtP_adap} and Thm.~\ref{th:BC} use SVD to transform the MIMO channel into a set of parallel streams and then distribute the transmit power and ADCs among them. Finding the optimal allocation of the transmit power and ADCs is equivalent to a mixed integer programming problem which is NP-hard \cite{bixby2004mixed}. In this paper, we perform waterfilling at the transmitter for the power allocation and perform exhaustive search over all ADC allocations at the receiver. A more complete analysis of different transmit power and ADC allocation methods suitable for a practical mmWave setting is provided in \cite{Abbas2020Thr}. 

For our first simulations, we use the results in Thm.~\ref{th:PtP_adap} to investigate the achievable rate of the adaptive threshold receiver in the single-user scenario. 
Fig. \ref{fig:PtP_ns} illustrates the achievable rate of a single-user MIMO communication system consisting of a user with $n_t = 16$, $n_r = 32$ and varying $n_q$. We use the upper bound $\min\{n_q, C\}$ also known as truncated Shannon capacity (TSC), where $C$ is the Shannon capacity of MIMO channel without quantization constraints. We observe that the proposed transmission scheme achieves the optimal high SNR rate of $n_q$ bits per channel-use, and that the performance is close to the optimal (TSC) in all SNRs.

Next, we plot the achievable rate of the adaptive threshold receiver in MIMO BC (Thm.~\ref{th:BC}). Fig. \ref{fig:BC_sim2} shows the achievable per user rate in two and three user MIMO BC communication systems, where the users are equipped with adaptive threshold receivers with $n_q = 16$. We compare the achievable rate with the truncated Shannon capacity with equal time shares, namely $\min\{n_q, C\}/n_u$, where $n_u$ is the number of users. We observe that the optimal high SNR achievable rate of $n_q = 16$ bits per channel-use is achieved for each user irrespective of the number of the users. Furthermore, we see that the proposed transmission scheme performs close to a BC system equipped with high resolution ADCs at the users, where the users employ TDMA with equal time-shares. This suggests that using adaptive threshold receiver with one-bit ADCs has the potential of reducing the power consumption with only a marginal degradation of the achievable rate of the system. 

An alternative strategy for BC is to use TDMA among users and then use the single-user MIMO  scheme as in Thm.~\ref{th:PtP_adap}. We call this approach naive TDMA. Compared to the scheme in Thm.~\ref{th:BC}, naive TDMA results in less energy consumption, since ADCs at each receiver are only active when the transmitter is sending data to the corresponding user. However, we can observe from Fig. \ref{fig:BC_sim2} that naive TDMA leads to a significant rate-loss compared to the proposed BC scheme in Thm.~\ref{th:BC}.

\section{Conclusion}
\label{sec:conclusion}
In this paper, we studied single-user and multiterminal communication scenarios over MIMO channels when a limited number of one-bit ADCs are available at the receiver terminals. We proposed two families of receivers, namely \textit{hybrid blockwise} and \textit{adaptive threshold}, which use a sequence of delay elements to allow for blockwise analog processing of the channel outputs. At high SNR, given a fixed number of one-bit ADCs, we have shown that the proposed receivers achieve the maximum transmission rate/rate region among all receivers with the same number of one-bit ADCs. An interesting direction for future work would be to also limit the number of low resolution DACs and optimize the transmitter and receiver architecture. Another direction is to find closed form expressions for the capacity of the hybrid one-shot and blockwise receivers for all SNRs.
\bibliography{main.bbl}

\appendices
\section{ Proof of Thm.~\ref{th:PtP}}
\label{app:th1}
To prove the theorem, we make use of the following proposition

\begin{Proposition}
\label{Prop:nchoosek}
Let $n\in \mathbb{N}$ and $\lambda\in (0,1)$ such that
\begin{align}
\label{eq:alpha}
    \frac{1}{2}\left(1-\sqrt{1-\frac{1}{3n}} \right)< \lambda
    < \frac{1}{2}\left(1+\sqrt{1-\frac{4}{12n+1}}\right),
\end{align}
Then, the following equality holds:
\begin{align}
\label{Eq:nchoosek}
\log{{n \choose \lambda n}}= n \mathrm{h_b}(\lambda)
-\frac{1}{2}\log{n}
-\frac{1}{2}\log{2\pi\lambda(1-\lambda)}
+O\left(1\right).
\end{align}
Particularly, for asymptotically large $n$, \eqref{Eq:nchoosek} holds for any fixed $\lambda\in (0,1)$. Furthermore, for $\lambda\in (0,\frac{1}{2})$, we have
\begin{align}
\label{Eq:nchoosek2}
\log{\sum_{i=0}^{\lambda n}{n \choose i}}= n \mathrm{h_b}(\lambda)
-\frac{1}{2}\log{n}
-\frac{1}{2}\log{\frac{2\pi\lambda(1-2\lambda)^2}{1-\lambda}}
+O\left(1\right),
\end{align}
where $\mathrm{h_b}(\cdot)$ is the binary entropy. Moreover, for a fixed $k$ and asymptotically large $n$
\begin{align}
\label{Eq:nchoosek3}
\log{\sum_{i=0}^{k}{n \choose i}}= k\log{n} +O\left(1\right),
\end{align}
\end{Proposition}
\begin{proof}
By Stirling's approximation, we have: 
\begin{align} 
\label{eq:sterling}
n! = {\sqrt {2\pi n}}\left({\frac {n}{e}}\right)^{n}\left(1+{\frac {1}{12n}}+O\left(\frac{1}{n^2}\right) \right).
\end{align}
Substituting \eqref{eq:sterling} in ${n \choose n\lambda} = \frac{n!}{(\lambda n)!((1-\lambda)n)!}$ and simplifying, we get:
\begin{align}
    &{n \choose n\lambda} = \frac{1}{\sqrt{2\pi n\lambda(1-\lambda)}}
     \frac{1}{\lambda^{n\lambda}(1-\lambda)^{n(1-\lambda)}}  \frac{1+ \frac{1}{12n} +O\left(\frac{1}{n^2}\right)}{\left(1+ \frac{1}{12n}\left( \frac{1}{1-\lambda} + \frac{1}{\lambda} \right) + O\left(\frac{1}{n^2}\right) \right)}\notag\\
    & \stackrel{\text{(a)}}{=} 
    \frac{1}{\sqrt{2\pi n\lambda(1-\lambda)}}  \frac{1}{\lambda^{n\lambda}(1-\lambda)^{n(1-\lambda)}}
   \left(1+ \frac{1}{12n} + O\left(\frac{1}{n^2}\right)\right) \left(1- \frac{1}{12n\lambda (1-\lambda)} + O\left(\frac{1}{n^2}\right)\right) \notag
    \\& = \frac{1}{\sqrt{2\pi n\lambda(1-\lambda)}}
    \frac{1}{\lambda^{n\lambda}(1-\lambda)^{n(1-\lambda)}}  \left(1+ \frac{1}{12n} \frac{\lambda(1-\lambda) - 1 }{\lambda(1-\lambda)} + O\left(\frac{1}{n^2}\right)\right),
\end{align}
where in (a), we have used the Taylor series expansion $\frac{1}{1+x} = 1 - x + O(x^2), x\in (-1,1)$. Note that \eqref{eq:alpha} ensures that $\frac{1}{12n}\left( \frac{1}{1-\lambda}+\frac{1}{\lambda}\right) \in (-1,1)$. Taking the logarithm:
\begin{align}
\label{eq:logp1}
&\begin{aligned}
    \log{{n \choose n\lambda}}&= 
    \frac{-1}{2}\log{{{2\pi n\lambda(1-\lambda)}}}
    - {n\lambda}\log{\lambda}-{n(1-\lambda)}\log{(1-\lambda)} 
    \\
    &+\log{
    \left(1+ \frac{1}{12n} \frac{\lambda(1-\lambda) - 1 }{\lambda(1-\lambda)} + O\left(\frac{1}{n^2}\right)\right)}.
\end{aligned}
\end{align}
Using the Taylor expansion $\ln{(1+x)}=1-x+O(x^2), x\in (-1,1]$, we have
\begin{align}
\label{eq:logp2}
    \log{
    \left(1+ \frac{1}{12n} \frac{\lambda(1-\lambda) - 1 }{\lambda(1-\lambda)} + O\left(\frac{1}{n^2}\right)\right)} =  1-  \frac{1}{12n} \frac{\lambda(1-\lambda) - 1 }{\lambda(1-\lambda)} + O\left(\frac{1}{n^2}\right).
\end{align}
Note that \eqref{eq:alpha} ensures that $\frac{1}{12n} \frac{\lambda(1-\lambda) - 1 }{\lambda(1-\lambda)} \in (-1,1]$. Substituting \eqref{eq:logp2} in \eqref{eq:logp1} and simplifying, we get \eqref{Eq:nchoosek}. 
To prove \eqref{Eq:nchoosek2}, consider the equality 
\begin{align}
\label{Eq:01}
    {n \choose \lambda n-j}= {n \choose \lambda n}\times
    \frac{\lambda n}{n-\lambda n +1}
    \times
    \frac{\lambda n-1}{n-\lambda n +2}\times 
    \cdots 
    \times
    \frac{\lambda n-j+1}{n-\lambda n +j}, 
    \lambda<\frac{1}{2}, j\in [\lambda n].
\end{align}
Note that
\begin{align}
    \label{Eq:11}
    \frac{\lambda n-i+1 }{n-\lambda n +i}\leq \frac{\lambda n}{n-\lambda n+i}<\frac{\lambda}{1-\lambda},\quad  i\in [\lambda n],
\end{align}
From \eqref{Eq:01} and \eqref{Eq:11}, we have:
\begin{align}
\label{Eq:31}
    \sum_{i=0}^{\lambda n}{n \choose \lambda n-i}\leq {n \choose \lambda n}\left(1+ \frac{\lambda}{1-\lambda}+\left(\frac{\lambda}{1-\lambda}\right)^2+\cdots \right)= {n \choose \lambda n} \frac{1-\lambda}{1-2\lambda}.
\end{align}
On the other hand, $\sum_{i=0}^{\lambda n}{n \choose \lambda n-i}\geq {n \choose \lambda n}$. As a result,
\begin{align}
\label{Eq:32}
    &\sum_{i=0}^{\lambda n}{n \choose \lambda n-i}= {n \choose \lambda n}O(1).
\end{align}
Taking the logarithm of \eqref{Eq:32} and substituting \eqref{Eq:nchoosek} proves \eqref{Eq:nchoosek2}. Equation  \eqref{Eq:nchoosek3} can also be proved in a similar fashion.
\end{proof}

To prove Thm.~\ref{th:PtP}, we only consider the case of non-zero thresholds. The case of zero thresholds is proved in a similar fashion. Denote $\munderbar{\mathbf{x}} \in \mathbb{R}^{\ell n_t}$ as concatenation of $\ell$ consecutive transmit signals, and $\munderbar{\mathbf{y}}$ as the concatenation of the corresponding received signals. We have 
\begin{align}
\label{eq:inout}
    \munderbar{\mathbf{y}} = (\mathbf{H}\otimes \mathbf{I}_{\ell} )\munderbar{\mathbf{x}} + \mathbf{\munderbar{\mathbf{n}}}
\end{align}
where $\mathbf{I}_{\ell}$ is the identity matrix of dimension $\ell \times \ell$ and $\mathbf{\munderbar{\mathbf{n}}} \in \mathbb{R}^{\ell n_r}$ is an independent noise vector.

Since, we consider SNR$\to \infty$, we design a strategy for $\mathbf{\munderbar{\mathbf{n}}} = \mathbf{0}$.
The channel output signal dimension for the noiseless system is $\ell\mathrm{rank}(\mathbf{H})$. We consider the pair $(\mathbf{V}\in \mathbb{R}^{\ell n_q \times \ell n_r},\mathbf{t} \in \mathbb{R}^{\ell n_q})$ which achieve the maximum number of decision regions in Prop.~\ref{Prop:Partition}. Therefore, at high SNR, the receiver can distinguish between ${\cal{R}}(\mathbf{H}\otimes \mathbf{I}_{\ell}, \mathbf{V},\mathbf{t})$ symbols and achieve the rate of $\log|{\cal{R}}(\mathbf{H}\otimes \mathbf{I}_{\ell}, \mathbf{V},\mathbf{t})|$. However, the receiver takes $\ell$ channel-uses for decoding each symbol and so the high SNR rate of the system becomes
\begin{align}
    R^{\mathrm{High~SNR}} &=\frac{1}{\ell} \log{\sum_{i=0}^{\ell\mathrm{rank}(\mathbf{H})} {\ell n_q \choose i}}
\end{align}

Next, we prove the converse result. Note that if we consider the noise vector $\mathbf{\munderbar{\mathbf{n}}}$ in \eqref{eq:inout} the dimension of the received signal space at the receiver becomes $\ell n_r$ instead of $\ell \mathrm{rank}(\mathbf{H})$. Therefore, using Prop.~\ref{Prop:Partition}, the number of regions and thus modulation symbols becomes $\sum_{i=0}^{\ell n_r} {\ell n_q \choose i}$, and we have
\begin{align}
  C^{\mathrm{High~SNR}}_{\ell} &\leq \frac{1}{\ell}\log{\sum_{i=0}^{\ell n_r} {\ell n_q \choose i}}
\end{align}

Using the asymptotic expansions derived in Prop.~\ref{Prop:nchoosek} to simplify the upper bound on the high SNR capacity and the achievable high SNR rate gives us \eqref{eq:linf} and \eqref{eq:nqinf}.  
\end{document}